\documentclass[11pt,a4paper,onecolumn]{IEEEtran}

\usepackage{amssymb,amsmath,amsfonts,color,mathrsfs,amsthm,verbatim,esint}
\usepackage{enumerate}
\usepackage{graphicx}
\usepackage[noend]{algpseudocode}

\usepackage{booktabs}
\usepackage{enumitem}
\usepackage{mathtools}
\usepackage{bbm}
\newtheorem{cor}{Corollary}[section]
\newtheorem{lemma}{Lemma}[section]
\newtheorem{thm}{Theorem}[section]
\newtheorem{prop}{Proposition}[section]

\newtheorem{definition}{Definition}[section]

\newtheorem{rem}{Remark}[section]

\DeclareMathOperator{\spn}{span}

\newcommand{\Q}{\mathbb{Q}}
\newcommand{\R}{\mathbb{R}}
\newcommand{\C}{\mathbb{C}}
\newcommand{\E}{\mathbb{E}}
\newcommand{\Z}{\mathbb{Z}}

\newcommand{\Span}{\mathrm{span}}

\newcommand{\diag}{\mathrm{diag}}
\newcommand{\diam}{\mathrm{diam}}

\newcommand{\vol}{\mathrm{Vol}}
\newcommand{\SL}{\mathrm{SL}}

\DeclareMathOperator{\ECDP}{ECDP}

\newcommand{\cami}[1]{\color{black}{#1}\color{black}}
\newcommand{\Tao}[1]{\color{black}{#1}\color{black}}

\begin{document}

\title{Well-Rounded Lattices: Towards Optimal Coset Codes for Gaussian and Fading Wiretap Channels}

\author{Mohamed Taoufiq Damir$^*$,\thanks{$^*$The first two authors contributed equally.} \thanks{M.~T. Damir and C.~Hollanti are with the Department of Mathematics and Systems Analysis, 
P.O. Box 11100, FI-00076,  
Aalto University, Finland. E-mails: \{mohamed.damir, camilla.hollanti\}@aalto.fi.} Alex Karrila$^*$\thanks{A.~Karrila is with the Institut des Hautes \'{E}tudes Scientifiques,
35 Route de Chartres, 91440 Bures-sur-Yvette, France. Email: karrila@ihes.fr; alex.karrila@gmail.com. 
}, Laia Amor\'os\thanks{L.~Amor\'os is with the Department of Computer Science, P.O. Box 11100, FI-00076,  
Aalto University, Finland. E-mail: laia.amoros@aalto.fi.}, Oliver Gnilke\thanks{O.~Gnilke is with the Department of Mathematical Sciences, University of Aalborg, Denmark. Email: owg@math.aau.dk.}, David Karpuk\thanks{D.~Karpuk was previously with the Department of Mathematics, Universidad de los Andes, Colombia. Email: davekarpuk@gmail.com.}, and Camilla Hollanti, \emph{Member, IEEE}\thanks{
Preliminary and partial results were presented in \cite{Gnilke, Barreal-Karrila-Karpuk-Hollanti,Gnilke-Barreal,spawc}. The present article is a unification and an extension of these works as well as the unpublished manuscript \cite{wiretap-old}.}
}

\maketitle

\begin{abstract}
The design of lattice coset codes for wiretap channels is considered. Bounds on the eavesdropper's correct decoding probability and information leakage are first revisited. From these bounds, it is  explicit that both the information leakage and error probability are controlled by the average flatness factor of the eavesdropper's lattice, which we further interpret geometrically. 
It is concluded that the minimization of the (average) flatness factor of the eavesdropper's lattice leads to the study of well-rounded lattices, which are shown to be among the optimal in order to achieve these minima. \cami{Constructions of some well-rounded lattices are also provided.}
\end{abstract}


\section{Introduction}
\label{sec:intro}

\subsection{Background}
In a wiretap channel Alice wishes to transmit information to Bob in the presence of an eavesdropper, Eve. A general objective of code design in a wiretap channel is to maximize the data rate and Bob's correct decoding probability while simultaneously minimizing  information leakage to Eve. It was shown in the seminal paper of Wyner \cite{Wyner} that the legitimate parties can design codes with asymptotically nonzero rate, zero error probability and zero information leakage. Today, this setup is particularly interesting in wireless channels that are open in nature.

As a practical construction of a wiretap code, \cite{Wyner-Ozarow} introduced the general technique of \emph{coset coding}, where random bits are added to the message to confuse the eavesdropper. In the specific case of a wireless channel, lattice codes are commonly used and the code lattice $\Lambda_b$ is endowed with a sublattice $\Lambda_e \subset \Lambda_b$, which carries random bits \cite{Oggier-Sole-Belfiore}.  In this work, we study non-asymptotic design criteria for $\Lambda_e$ to maximize the security of a lattice wiretap code for a fixed dimension. We will tacitly assume throughout that $\Lambda_e$ is a sublattice of a fixed reliable code lattice $\Lambda_b$ with fixed nesting index $[\Lambda_b : \Lambda_e]$.

\subsection{Related Work}
\label{subsec:related}


The security of lattice coset codes is often  measured either by Eve's correct decision probability (ECDP), or alternatively by the mutual information of the message and Eve's received signal. For the \emph{additive white Gaussian noise} (AWGN) channel, upper bounds are known for both approaches \cite{Oggier-Sole-Belfiore, Viterbo-flatness}. More importantly, both are increasing functions of the \emph{flatness factor} of the lattice $\Lambda_e$, yielding its minimization as a design criterion. In \cite{luzzi},   the secrecy criterion  based on the flatness factor in \cite{Viterbo-flatness} is extended to the case of multiple-input multiple-output (MIMO) channels.
Sequences of lattice coset codes achieving security and reliability are also constructed in \cite{luzzi}.

Codes achieving AWGN channel capacity using the flatness factor are proposed in \cite{Ling-Belfiore}. More recently some of the extra conditions needed in \cite{Ling-Belfiore} are removed in \cite{Campello-Dadush-Ling}, establishing a direct link between AWGN-goodness and capacity-achieving codes at the cost of using dithering. Polar lattices achieving AWGN channel capacity for any signal-to-noise ratio (SNR) are given in \cite{polar-lattices}.

For different fading channel models, various alternative design criteria based on error probability and information leakage  bounds were derived for both single-input single-output (SISO) and MIMO channels in \cite{luzzi,Belfiore-Oggier, belfiore_mimo}, and \cite{mirghasemi}. 
\cami{In the pioneering work \cite{Belfiore-Oggier}, the authors derived the so-called \emph{inverse norm sum} (cf. \cite[Sec. III-B]{Belfiore-Oggier}) as a design criterion for the fading wiretap channel. To this end, they assume that eavesdroppers signal quality is relatively high, so that certain terms in the eavesdropper's correct decoding probability can be ignored or simplified. Subsequently, the resulting inverse norm sum was studied in, \emph{e.g.}, \cite{INS_me, INS_roope}. This inverse norm sum form naturally leads to the study of Dedekind zeta functions as the related series is then taken over the field norms. In this paper, we will take a step back and utilize an earlier derivation step (see \cite[eq. (14)]{Belfiore-Oggier}) as our starting point, hence avoiding making any assumption on the eavesdropper's signal quality other than being worse than that of the legitimate receiver's. This results in  an ECDP upper bound that is valid for any signal-to-noise ratio (SNR) and is asymptotically tight at low SNR. With this derivation, the related series is taken over Euclidean norms, therefore naturally leading to a study on the flatness factor and Epstein zeta functions.} 



For the  SISO and MIMO fading wiretap channel the authors in \cite{luzzi}  establish a design criterion for fading and MIMO channels from an error probability perspective aiming at asymptotic and universal goodness. They suggest that the normalized product distance (SISO) and the normalized minimum determinant (MIMO) of the faded lattice and its dual should be maximized simultaneously. Further, they proposed an algebraic construction of lattices achieving strong secrecy and semantic security for all secrecy rates $R < C_b- C_e -c$, where $C_b$ and $C_e$ are Bob's and Eve's channel capacities respectively, and $c$ a constant depending on the lattice.  Unfortunately, for such constructions the value of $c$ turns out to be relatively large, and the explicit construction is problematic (see the discussion in \cite[Sec. VIII]{luzzi} for more details).


To the best of our knowledge, constructive, non-asymptotic design criteria as well as practical and explicit low-dimensional code constructions remain an open problem for the fading SISO and MIMO wiretap channels.

\Tao{
In the present paper we propose the set of well-rounded lattices as a search space for secure lattice codes over wiretap channels. 
Well-rounded lattices appear in various arithmetic and geometric problems. 
For example, a classical theorem due to Voronoi \cite{Voronoi} implies that the local maxima of the sphere packing function are all realized at well-rounded lattices. 
In \cite{mcmullen}, well-rounded lattices have been investigated in the context of Minkowski conjecture. Furthermore, topological properties of the set of well-rounded lattices have also been of interest. For instance, in \cite{ash}, it was proved that the space of all (determinant one) lattices retracts to the space of well-rounded lattices. More recently, a result in \cite{solan2019stable} states that for any lattice $\Lambda$ there exists a determinant one diagonal real matrix $a$ with positive entries such that $a\cdot \Lambda$ is a well-rounded lattice. This result has been used in \cite{TaoLenny} to show that the problem of maximizing the minimum product distance can be restricted to the set of well-rounded lattices without loss of generality.

It is also worth mentioning that most of the lattice-based cryptographic protocols rely on the hardness of the shortest vector problem. On the other hand, the problem of determining all the successive minima of an arbitrary lattice is believed to be strictly harder \cite{micciancio}. However, if the lattice is well-rounded, these two problems become equivalent.

}

\subsection{Contributions}
In this paper, we first motivate a natural and simple lattice design criterion. We take a general channel model with linear fading and Gaussian noise. We recall the strategy used to derive probability bounds for Rayleigh fading SISO and MIMO channels in \cite{Belfiore-Oggier, belfiore_mimo}, and give a slight variant of the information leakage bounds derived in \cite{luzzi_isit16,mirghasemi}. We also derive an information leakage bound in the so-called mod $\Lambda_s$ channel. In particular, we obtain all bounds explicitly as increasing functions of the \textit{average flatness factor} of $\Lambda_e$, and these reduce to the probability bounds of \cite{Belfiore-Oggier} and \cite{belfiore_mimo} in Rayleigh fading channels. The equivalence of the error probability and information leakage bounds
 as presented here hopefully clarifies the situation where several alternative design criteria for each different channel model have been derived with occasionally rather complicated analytic expressions. Some of these results are recovered from \cite{wiretap-old}. 

Having motivated analytic design criterion, we study practical finite-dimensional lattice designs. For AWGN and Rayleigh fast fading channels, we motivate a geometric approach to minimize the (average) flatness factor.
This is done by using some results from the theory of Epstein zeta functions and spherical designs. It is shown that we can restrict the minimization problem to the family of \emph{well-rounded lattices}. This provides a constructive criterion for the search of well-performing lattice coset codes in a fixed dimension (in contrast to asymptotic approaches from the capacity point of view). A collection of simulations are presented to give further credit to well-rounded lattices. 

\cami{While we intend this paper to mainly serve as a proof-of-concept for well-rounded lattices in communications, we also provide some new constructions of well-rounded lattices. To this end, we point out that well-rounded lattices are of measure zero among the space of all lattices of fixed dimension $n$. In fact, a lattice is well-rounded whenever its successive minima are all equal. Consequently, the space of well-rounded lattices is defined by $n-1$ equalities. Thus, this space is of co-dimension $n-1$. Note that this argument implies that the space of well-rounded lattices is not full-dimensional in the space of all lattices and  hence has vanishing (Haar) measure. While well-rounded lattices may seem very rare from a probabilistic point of view, the dimensionality argument (co-dimension $n-1$) shows that the space of well-rounded lattices is indeed large enough to meaningfully construct such lattices;  in \cite{costawell}, infinitely many non-similar well-rounded lattices in every prime dimension $p$  were constructed. In \cite{damir2020bases}, the authors constructed well-rounded lattices in all dimensions and showed the existence of infinitely many non-similar well-rounded lattices in every dimension of the form $n=s^2+2$, where $s$ is a positive odd integer such that $s^2+1$ is square-free. In this paper, we contribute to this study by constructing well-rounded lattices over some real number fields of dimension $\varphi(m)$, where $\varphi$ is the Euler totient function and $m>1$ an odd integer.}

It is worth mentioning that some previous design criteria could also be restricted to the space of well-rounded lattices. For example in \cite{luzzi} the authors proposed the maximization of the normalized product distance of the lattice and its dual in the case of the wiretap fading channel; and the maximization of the packing density of the lattice and its dual for the Gaussian wiretap channel. 
Indeed, it is well known that the well-roundedness property is a necessary condition for a lattice to achieve a local maximum for the packing density function. \cami{Namely, the local maxima are achieved by extreme lattices, and by a classical resulf of Voronoi a lattice is known to be extreme if and only if it is perfect and eutactic \cite{Martinet_perfect,Voronoi}.  It is easy to see that perfect lattices are well-rounded \cite[Prop 3.3]{Schurmann}.} Furthermore, it was proved in \cite{TaoLenny} that the maximization of the normalized product distance can be restricted to the space of well-rounded lattices without loss of generality. This provides nice coherence to the asymptotic ($n\rightarrow \infty$) results in \cite{luzzi} and the results concerning any fixed dimension in this paper. 


Compared to the earlier conference publications \cite{Gnilke, Barreal-Karrila-Karpuk-Hollanti,Gnilke-Barreal,spawc} preceding this paper, we provide more rigorous justifications for the choice of well-rounded lattices \cami{as well as a new construction method} (\emph{e.g.}, in Section \ref{sec: WR}). Moreover, we hope to unify all the previous and new results here in a comprehensive way.

\subsection{Organization}

This paper is organized as follows. In Section \ref{sec: preli} we give the mathematical preliminaries.
Section~\ref{sec: system} introduces the channel models, lattice coset codes, and the detailed setups for the error probability and information leakage bounds.
In Section~\ref{sec: AWGN} we give the eavesdropper's information leakage and correct decoding probability bounds in the AWGN channel, yielding the average flatness factor analytic design criterion.
In Section \ref{sec: fading channel} we generalize the previous bounds to an arbitrary fading channel.
Section \ref{sec: WR} shows the importance of well-rounded lattices regarding the behaviour of the flatness factor \cami{and presents a new construction.}
In Section \ref{sec: sim} we show some simulations of a wiretap setting with the aim to back-up the design criterion deduced earlier. We first compare constructions from well-rounded and from non-well-rounded lattices. Then we proceed to compare different well-rounded lattices, showing how the choice among them is not necessarily straightforward and warrants interesting further research.

\section{Mathematical preliminaries}
\label{sec: preli}

\subsection{Information-theoretic definitions}
\label{subsec: Information-theoretic definitions}

We consider a message as a random variable $\mathbf{m}$ taken from a finite message set $\mathcal{M}$ according to some distribution, and denote by $\mathbf{y}$ the random variable representing the output of the channel at the eavesdropper's end. The entropy $H [\mathbf{m}]$, conditional entropy $H [\mathbf{m} | \mathbf{y}]$ and mutual information $I [\mathbf{m} ; \mathbf{y}]$ are defined as usual, see \textit{e.g.} \cite{Shamai}. We recall the trivial bounds $0 \le I [\mathbf{m} ; \mathbf{y}] \le H [\mathbf{m}] \le \log | \mathcal{M} |$.  In this paper we are interested in minimizing the mutual information $I [ \mathbf{m}; (\mathbf{y}, \mathbf{h}) ] $ where $\mathbf{h}$ is the channel state, known by the receiver. The random variables $\mathbf{m}$ and $\mathbf{h}$ are assumed independent, and $\mathbf{y} $ depends on $\mathbf{m}$, $\mathbf{h}$ and additionally a random variable describing noise. One has
\begin{equation}
\label{eq: information with double condition}
I [ \mathbf{m}; (\mathbf{y}, \mathbf{h}) ] = \mathbb{E}_{\mathbf{h}} \{ I [ \mathbf{m} ; \{ \mathbf{y} | \mathbf{h} \} ] \},
\end{equation}
In other words, the quantity of interest is obtained as the expectation of the mutual information over the different channel states. This will be our strategy to compute information leakage bounds.  Throughout this paper, in order to streamline the presentation, we will often use the same letter to denote a random variable and a realization of that same random variable, with the meaning always clear from the context.

\subsection{Lattices}

\subsubsection{Basic concepts}
     \cami{ Throughout the paper, a \emph{lattice} $\Lambda$ is a discrete Abelian subgroup of a real vector space. Let $\{\mathbf{b}_1,\ldots, \mathbf{b}_k\}$ be linearly independent in $\R^n$. Then we can write a lattice in terms of its basis $\{\mathbf{b}_i\}$ as $\Lambda=\{\sum_{i=1}^{k} z_i \mathbf{b}_i\,:\, z_i\in\mathbb{Z},\, k\leq n\}\subset\R^n$. Here, $k$ and $n$  are referred to as the \emph{rank} and \emph{dimension} of the lattice, respectively, and the lattice is said to be \emph{full (rank)} if $k=n$. A lattice generator matrix  $M_\Lambda =(\mathbf{b}_i)_{1\leq i\leq k}\subseteq \R^{n\times k}$ consists of the lattice basis vectors as columns. The \emph{volume} of the lattice is defined as the volume of the fundamental parallelotope, and can be computed as $\vol(\Lambda)=\sqrt{\det(M_\Lambda^tM_\Lambda)}$, where $t$ denotes the transpose.} The dual lattice is denoted by $\Lambda^*$, which by definition has generator matrix $(M_\Lambda^{t})^{-1}$.  The Voronoi cell of $\Lambda$ centered at $\mathbf{0}$ is defined by $\mathcal{V}(\Lambda) =\{\mathbf{x}\in\R^n\,:\, \forall \mathbf{t}\in\Lambda\,\, ||\mathbf{x}||<||\mathbf{x}-\mathbf{t}||\} $,  where $||\mathbf{x} ||=\sqrt{\sum_{i=1}^n x_i^2}$ is the euclidean norm of $\mathbf{x}$. A lattice $\Lambda$ has \textit{full diversity} if for all nonzero $\mathbf{t} \in \Lambda$, all the components $t_i$ are nonzero. Unless stated otherwise, all lattices we consider will be full.

    A vector $\mathbf{x} \ne \mathbf{0}$ of a lattice $\Lambda$ is a \textit{minimal (length) vector} if it is of minimal length among all nonzero lattice vectors. The \textit{minimal norm}  of $\Lambda$ is then $ \lambda_1(\Lambda)=\min_{0\neq \mathbf{x}\in\Lambda}||\mathbf{x}||$. 
    The problem of finding the lattice with maximal $\lambda_1(\Lambda)$ among all lattices of unit volume is \textit{the (lattice) sphere packing problem}. The best sphere packings are known in low dimensions \cite{Sloane, Blichfield, Kumar}, but in general the sphere packing problem is hard.

\subsubsection{Well-rounded lattices}
    A special class of lattices that will be of our interest are the so-called well-rounded lattices.

\begin{definition}
	Let $\Lambda \subset \R^n$ be a lattice and let $S(\Lambda):=\{ \boldsymbol{t} \in \Lambda \; :\; || \boldsymbol{t}|| =\lambda_1(\Lambda) \}$ be the set of shortest vectors. Then we say that $\Lambda$ is \emph{well-rounded (WR)} if $\spn_\R(\Lambda)=\spn_\R(S(\Lambda))$.
\end{definition}

In other words, the set of minimal vectors spans a vector space that has the same dimension as the span of the whole lattice, and thus the vector spaces coincide. The set of minimal length vectors $S(\Lambda)$ does not necessarily form a basis for $\Lambda$ {\cite[Chapter 2]{PhongNguyen}}. They are known to form a basis for all $n \le 4$ as mentioned in \cite{Martinetbasis}. An equivalent definition uses the successive minima of a lattice to define well-roundedness.

\begin{definition}
	Let $\Lambda \subset \R^n$ be a lattice of rank $k\leq n$, and let $\lambda_i = \lambda_{i}(\Lambda) := \inf \left\{\left. r \right| \dim(\Span(\Lambda \cap \mathcal{B}_r))\geq i \right\} $ be the successive minima of $\Lambda$, where $\mathcal{B}_r$ is the closed ball of radius $r$ around the origin. Then $\Lambda$ is called \emph{well-rounded} if $\lambda_1 = \lambda_2= \cdots = \lambda_k$.
\end{definition}

A subclass of well-rounded lattices is given by the so called generic well-rounded lattices.

\begin{definition}
	Let $\Lambda \subset \R^n$ be a lattice of rank $k\leq n$. We say that $\Lambda$ is \emph{generic well-rounded} if it is well-rounded \cami{and has exactly $\kappa(\Lambda)=|S(\Lambda)|=2k$ shortest vectors, \emph{i.e.}, its \emph{kissing number} $\kappa$ coincides with  $\kappa(\mathbb{Z}^k$).}
\end{definition}

\begin{rem}\label{rem:kissing-discussion} \cami{While the densest lattice packing is solely determined by the shortest vector, for the theta series and consequently flatness factor minimization problem the kissing number also plays a role. Ideally, we would like to simultaneously maximize the shortest vector length \textbf{and} minimize the number of them. This seems to be hard problem; for instance, the lattice $\mathbb{Z}^n$ is well-known to be the worst among WR lattices in terms of the packing density ($\lambda_1=1$). On the other hand, it has the smallest possible kissing number $2n$ among WR lattices. More generally, the densest WR lattices provide the optimal solution in terms of the shortest vector, while the generic WR lattices provide the lowest kissing numbers among WR lattices. However, no obvious tradeoff between the two extremes is known.  }
\end{rem}

\begin{definition}
 \cami{   A lattice is called \emph{perfect}, if it is completely determined by the set $S(\Lambda)$, \emph{i.e.}, there is only one positive definite quadratic form taking value $1$ at all points of $S(\Lambda)$.}
\end{definition}
\cami{It is easy to see that perfect lattices are WR \cite[Prop. 3.3]{Schurmann}. As mentioned before, extreme lattices give the local minima of the sphere packing density, and a lattice is extreme if and only if it is perfect and eutactic \cite{Voronoi}, hence also WR. 
}

Well-rounded lattices have been previously proposed \cite{Gnilke,Barreal-Karrila-Karpuk-Hollanti,Gnilke-Barreal} as good lattices for coset codes (introduced later), based on various performance measures (ECDP, flatness factor, mutual information)  \cite{Viterbo-flatness,luzzi_isit16,luzzi,Belfiore-Oggier,mirghasemi}, and have been shown to outperform non-well-rounded ones.
It is possible to find all WR sublattices of a given lattice and a given index by searching through all possible combinations of vectors of suitable length. For instance, the WR non-orthogonal lattices used in our simulations in Section \ref{sec: sim} were found after a few minutes of randomly testing combinations of integer vectors.
Nevertheless, theoretical results to construct a variety of good (generic) well-rounded lattices is an open problem that has become of great interest \cite{costawell,damir2020bases}. 

\subsubsection{Gaussian sums, theta function, and flatness factor}

We denote the $n$-dimensional Gaussian zero-mean probability density function (PDF) with variance $\sigma^2$ by
\begin{equation}
g_n (\mathbf{t}; \sigma) = \frac{1}{(\sqrt{2 \pi} \sigma)^n} \exp \left(- \frac{||\mathbf{t}||^2}{2 \sigma^2} \right), \quad \mathbf{t}\in\Lambda.
\end{equation}
and its (possibly shifted) lattice sums by
\begin{equation}
g_n (\Lambda + \mathbf{x} ; \sigma) := \sum_{\boldsymbol{\lambda} \in \Lambda} g_n (\boldsymbol{\lambda} + \mathbf{x} ; \sigma),\quad \mathbf{x}\in\R^n.
\end{equation}

\cami{One can associate to any (positive-definite) lattice  a \emph{theta function} given by
$$\Theta_{\Lambda} (q) = \sum_{\boldsymbol{\lambda} \in \Lambda } q^{|| \boldsymbol{\lambda} ||^2},$$ where $q=e^{i\pi \tau}$ and $\mathrm {Im} \,\tau >0$. The theta function of a lattice is then a holomorphic function on the upper half-plane. In this paper, we only consider real theta functions and typically choose $q=e^{-1/2 \sigma^2}$, since $g_n (\Lambda; \sigma) = \frac{1}{(\sqrt{2 \pi} \sigma)^n} \Theta_{\Lambda} (e^{-1/2 \sigma^2})$.



Then, we can also write it as the generating function
\begin{equation}
\label{eq: theta function definition}
\Theta_{\Lambda} (q) = \sum_{\boldsymbol{\lambda} \in \Lambda } q^{|| \boldsymbol{\lambda} ||^2} = 1 + \# \{ \boldsymbol{\lambda} \in \Lambda \ : \  || \boldsymbol{\lambda} ||^2 = \lambda_1(\Lambda)^2 \} q^{ \lambda_1(\Lambda)^2 } + \ldots
\end{equation}
where $| q | < 1$, and the series converges absolutely. 
}

It is easy to see that $g_n (\Lambda + \mathbf{x} ; \sigma) $ is $\Lambda$-periodic as a function of $\mathbf{x}$ and it defines a PDF on $\mathcal{V}(\Lambda)$, called the \textit{lattice Gaussian} PDF. The deviation of the lattice Gaussian PDF from the uniform distribution on $\mathcal{V}(\Lambda)$ is characterized by the \textit{flatness factor} $\varepsilon_{\Lambda} (\sigma)$, which we define for full lattices by
\begin{equation}
\label{eq: flatness factor definition}
\varepsilon_{\Lambda} (\sigma) : = \max_{\mathbf{u} \in \R^n} \left| \frac{g_{n}(\Lambda + \mathbf{u}; \sigma)}{1/\vol(\Lambda)} - 1 \right|,
\end{equation}
where we can maximize over $\R^n$ by periodicity. 
\cami{
By \emph{average flatness factor} we mean 
\begin{equation}
\label{eq:averageFF}
\E [ \varepsilon_{ \mathbf{h}  \Lambda_e}(\sigma) ],
\end{equation}
where the expectation is taken over different fading realizations $\mathbf{h}$.
}

The flatness factor was introduced as a wiretap information theory tool in \cite{Viterbo-flatness}. The Poisson summation formula yields the following useful equalities:
\begin{equation}
\label{eq: dual theta flatness factor formula}
\varepsilon_{\Lambda} (\sigma) 
=  \vol(\Lambda) g_{n} (\Lambda; \sigma) - 1 
= \frac{\vol(\Lambda)}{(\sqrt{2 \pi} \sigma)^n} \Theta_{\Lambda} (e^{-1/2 \sigma^2}) - 1 
= \Theta_{\Lambda^*} (e^{- 2 \pi \sigma^2}) - 1.
\end{equation}
From the last expression it is clear that the flatness factor is strictly decreasing in $\sigma$ and tends to zero as $\sigma \to \infty$. It also implies the scaling property $\varepsilon_{a \Lambda} (a \sigma) = \varepsilon_{\Lambda} (\sigma)$. We remark that if a non-full lattice $\Lambda$ is generated by $M_\Lambda \in \R^{n \times m}$, then $\varepsilon_{\Lambda} = \varepsilon_{\sqrt{M_\Lambda^t M_\Lambda}}$, where the matrix square root $Q = \sqrt{M_\Lambda^t M_\Lambda} \in \R^{m \times m}$ satisfies $Q^t Q = M_\Lambda^t M_\Lambda$.

We define the \textit{variational distance} $V(\rho,q)$ of two PDFs $\rho$ and $q$ as
\begin{equation}
\label{var}
V(\rho,q) = \int_{\R^n} | \rho (\mathbf{y}) - q(\mathbf{y}) |\ d \mathbf{y}.
\end{equation}
It is clear that the flatness factor bounds the variational distance of the lattice Gaussian distribution on $\mathcal{V} (\Lambda)$ and the uniform distribution on $\mathcal{V} (\Lambda)$,
\begin{equation}
V(g_n (\Lambda + \mathbf{y} ; \sigma)|_{\mathcal{V}(\Lambda)} 1/ \vol(\mathcal{V}(\Lambda))|_{\mathcal{V}(\Lambda)}) \le \varepsilon_{\Lambda} (\sigma) .
\end{equation}
The connection between the flatness factor and information leakage estimates is now illustrated by the following lemma that is crucial both in the estimates of \cite{Viterbo-flatness} and in this paper.

\begin{lemma}
\label{lemma: ff and information} \cite[Lemma 2]{Viterbo-flatness}
Let $\mathbf{y}$ be an $\R^n$-valued random variable, and let $\mathbf{m}$ have any distribution on a message set $\mathcal{M}$ such that $| \mathcal{M} | \ge 4$. Denote by $\rho_{\mathbf{y}| \mathbf{m}}$ the PDF of $\mathbf{y}$ given a message realization. Suppose that there exists some PDF $q$ on $\R^n$ such that, for all message realizations, we have
$V(\rho_{\mathbf{y}| \mathbf{m}},q) \le \varepsilon \le 1/2 $. Then we have
\[
I[\mathbf{m}; \mathbf{y}] \le 2 \varepsilon \log | \mathcal{M} | - 2 \varepsilon \log(2 \varepsilon ) 
\]
\end{lemma}

We remark that the assumption $\varepsilon \le 1/2$ is implicit in \cite{Viterbo-flatness}, where the authors are interested in sequences of codes where $I[\mathbf{m}; \mathbf{y}] \to 0$. It is however necessary, as seen by taking $\varepsilon \to \infty$. We note that the above bound achieves the trivial bound $I[\mathbf{m}; \mathbf{y}] \le \log | \mathcal{M} |$ at $\varepsilon = 1/2$.


\section{System model}\label{sec: system}

\subsection{Channel models}

\subsubsection{General wireless channel with fading and noise}\label{general}

We consider in this paper several models representing a \textit{wireless channel model with (linear) fading and (additive white Gaussian) noise}. That is, if Alice transmits a vector $\mathbf{x} \in \R^n$, the receiver observes a vector
\begin{equation*}
    \mathbf{y} = \mathbf{h} \mathbf{x}+\mathbf{n},
\end{equation*}
where $\mathbf{h} \in \R^{m \times n }$,  $m \geq n$, is the channel fading matrix, and $\mathbf{n} \in \R^m$ the noise vector, composed of i.i.d. components $n_i \sim \mathcal{N}(0,\sigma^2)$. In all our considerations, the random variables $\mathbf{h}$, $\mathbf{x}$ and $\mathbf{n}$ are assumed independent, and $\mathbf{h}^t \mathbf{h} \in \R^{n \times n}$ is assumed to be full rank with high probability.

Both receivers, Bob and Eve, are assumed to have perfect channel state information (CSIR), \emph{i.e.}, know the realized value of $\mathbf{h}$. The transmitter is only assumed to know the channel statistics. The theoretical analysis in this paper only considers the channel between Alice and Eve, and we consequently forgo subscripts specifying the receiver. The assumptions on Bob's channel come into play through the choice of typical code lattices $\Lambda_b$ in the examples.

The information leakage bound derivations in this paper concern the general wireless channel model with linear fading and additive white Gaussian noise. 

\subsubsection{Important special cases}
\label{subsec: MIMO model}
\label{subsec: SISO model}

The general wireless channel model with fading and noise does not fix the distribution of the fading matrix $\mathbf{h}$. Different choices of this distribution will yield several classical channel models, some of which we exemplify here. We refer the reader to \cite{Oggier-Sole-Belfiore, belfiore_mimo, Viterbo} for background on these models.

The \textit{additive white Gaussian noise (AWGN) channel} is obtained when $\mathbf{h}$ is deterministically the identity matrix $I_n \in \R^{n \times n}$. AWGN wiretap channels have been studied in, \emph{e.g.}, \cite{Oggier-Sole-Belfiore}.

The \textit{Rayleigh fast fading SISO channel} is obtained by choosing $\mathbf{h}=\diag(h_1,\ldots,h_n) \in \R^{n \times n}$ a square diagonal matrix, where the diagonal entries are i.i.d. Rayleigh distributed with parameter $\sigma_{h}$. We recall that this distribution is characterized by the PDF
\begin{equation}
\label{PDF}
r(h) = \frac{h}{\sigma_{h}} \exp \left( - \frac{h^2}{2 \sigma_{h}^2} \right)
\end{equation}
on the positive real line $h \in \R_{ \geq 0}$. The \textit{Rayleigh block fading SISO channel} is similar to the fast fading case, except that the diagonal entries are not all independent. Instead, $n$ is assumed divisible by some integer $L$, $n=m \times L$, and the diagonal entries of $\mathbf{h}$ then consist of $m$ i.i.d. Rayleigh random variables, each repeated $L$ times. 
Rayleigh block and fast fading wiretap channels have been studied in \cite{Belfiore-Oggier,Oggier-Sole-Belfiore}, with a practical USRP implementation in \cite{Oggier_slow}.

The \textit{quasi-static Rayleigh fading MIMO channel} is obtained when $\mathbf{h} \in \R^{m \times n }$, with $m$ and $n$ both even, is the real matrix corresponding to a twice smaller complex matrix $\tilde{\mathbf{h}} \in \C^{m/2 \times n/2 }$ whose entries are i.i.d. complex Gaussian, $\tilde{h}_{i,j} \sim \mathcal{N}_{\C} (0, \sigma_h^2)$. That is, the real and imaginary parts $a_{i,j}$ and $b_{i,j}$ of $\tilde{h}_{i,j}$ are independent $\mathcal{N} (0, \sigma_h^2/2)$ real Gaussians, and $\mathbf{h}$ consists of $2 \times 2$ blocks
\begin{equation*}
\small{
\left(
\begin{array}{cc}
a & -b \\
b & a
\end{array}
\right).
}
\end{equation*}
Lattice code design criteria for MIMO wiretap channels have been considered in \cite{belfiore_mimo,luzzi,mirghasemi}.




\subsection{Coset coding}
\label{subsec:coset coding}

\subsubsection{Basic concepts}

Coset coding was first proposed by Ozarow and Wyner in \cite{Wyner-Ozarow} and adapted to the lattice coding setting in \cite{Oggier-Sole-Belfiore}. In lattice coset coding, Alice and Bob use nested lattices $\Lambda_e \subset \Lambda_b \subset \R^n$. For a message space $\mathcal{M}$ with $|\mathcal{M}| = [\Lambda_b : \Lambda_e]$, we have a fixed injective map $\mathcal{M} \to \Lambda_b \cap \mathcal{V} (\Lambda_e)$. To transmit a message $\mathbf{m}$, Alice first maps it by this injection to $\boldsymbol{\lambda}_{\mathbf{m}} $. Then, she transmits a vector $\mathbf{x}$ which is a random representative of the $[\Lambda_b : \Lambda_e]$ equivalence class of $\boldsymbol{\lambda}_{\mathbf{m}}$. Alice's transmitted vector can thus be written as
\begin{equation}
\label{representant}
\boldsymbol{\lambda}_{\mathbf{m}} + \boldsymbol{\lambda} \in \boldsymbol{\lambda}_{\mathbf{m}} + \Lambda_e \in \Lambda_b / \Lambda_e,
\end{equation}
where $\boldsymbol{\lambda} \in \Lambda_e$ is a random vector encoding the choice of the representative inside the $[\Lambda_b : \Lambda_e]$ equivalence class. Different distributions of $\boldsymbol{\lambda}$ yield different variants of coset coding. Note that $\boldsymbol{\lambda}$ need not be independent of $\mathbf{m}$\footnote{The vectors $\boldsymbol{\lambda}$ may also encode and hence be determined by public messages $\mathbf{m}_{pub}$. In this case one would typically assume $\boldsymbol{\lambda}$ independent of Bob's private message $\mathbf{m}$ and uniform over some subset of $\Lambda_e$.}.

The rate is $R = \frac{1}{n}\log_2 |\mathcal{M}|$ bits per real channel use.  The rate can be divided into two parts, the one related to the message and the other to the random bits:
\begin{equation}
R = R_m + R_r.
\end{equation}
Near-complete secrecy is then achieved even if Eve receives information with a nonzero rate, approximately equal to $R_r$. Conceptually, Eve can decode the ``coarse'' lattice $\Lambda_e$, which contains only random bits, but not the ``fine'' lattice, which contains the actual information.

If $\mathbf{x} = M_\Lambda\mathbf{z}$ for $\mathbf{z} \in \Z^n$, then $\mathbf{h} \mathbf{x} = \mathbf{h} M_\Lambda\mathbf{z}$, and we can think of a lattice code under fading with CSIR as a Gaussian-channel lattice code where the code lattice realizes a random lattice with generator matrix $\mathbf{h} M_\Lambda$. We will henceforth denote the faded lattices $\Lambda_b$ and $\Lambda_e$ by $\Lambda_{b,\mathbf{h}}$ and $\Lambda_{e,\mathbf{h}}$, respectively.

\subsubsection{Setups and bounds}
\label{subsubsec: shaping lattices}

The nature of the boundary of the transmission region in a lattice coset code is of utmost importance in deriving probability and mutual information estimates.  We consider three different estimates, arising by neglecting the boundary, removing it by a modulo operation of a shaping lattice $\Lambda_s$, and smoothing the boundary, respectively. The estimates are for:

\begin{itemize}
\item[] \textbf{Case 1:} The eavesdropper's correct decoding probability (ECDP), assuming that she decodes to the closest point of $\Lambda_b$. The same ECDP estimate holds for closest-point decoding in the mod $\Lambda_s$ channel discussed below.
\item[] \textbf{Case 2:} The eavesdropper's information leakage, assuming that she has the mod $\Lambda_s$ channel\footnote{\cami{The mod $\Lambda_s$ channel with the shaping lattice is artificial; Eve only receives information about the equivalence class of the received vector modulo the shaping lattice $\Lambda_s$. In the computations, this merely plays the technical role of removing boundary effects when Alice uses the uniform representative strategy in $\Lambda_s$, so that we can derive an explicit information bound.}} and Alice chooses uniform random representatives of the coset classes.
\item[] \textbf{Case 3:} The eavesdropper's information assuming that Alice uses Gaussian coset coding, also discussed below.
\end{itemize}


\subsubsection{The mod $\Lambda_s$ channel and uniform random representatives}

The following shaping lattice approach is identical to \cite{Viterbo-flatness}, called the mod $\Lambda_s$ channel: take three nested lattices $\Lambda_s \subset \Lambda_e \subset \Lambda_b \subset \R^n$ called shaping, coset, and code lattice, respectively. Then, the random part $\boldsymbol{\lambda}$ described in the general coset coding strategy has a uniform distribution on the $[\Lambda_e : \Lambda_s]$ representatives of $\Lambda_e / \Lambda_s$ in $\mathcal{V} (\Lambda_s )$. This is called the uniform representative strategy.  The physical message received by Eve is $\mathbf{y}$ as in the channel equation, but in the mod $\Lambda_s$ channel, Eve only receives knowledge of the equivalence class $\mathbf{y} / \Lambda_{s, \mathbf{h}}$. 

\subsubsection{Discrete Gaussian coset coding}

In the discrete Gaussian coding, the boundary effects of the transmission region are handled by smoothing the boundary. Fixing a message $\mathbf{m}\in\mathcal{M}$, the random part $\boldsymbol{\lambda}$ of the message $\mathbf{m} = \boldsymbol{\lambda}_\mathbf{m} + \boldsymbol{\lambda}$ is chosen so that the transmitted vector $\mathbf{x}\in \boldsymbol{\lambda}_\mathbf{m} + \Lambda_e$ has the centered discrete Gaussian distribution on the shifted lattice $\boldsymbol{\lambda}_\mathbf{m}+\Lambda_e$
\begin{equation}
\label{eq: discrete Gaussian probabilities}
P(\mathbf{x} ) = g_n (\mathbf{x}; \sigma_s) / g_n (\Lambda_e + \boldsymbol{\lambda}_\mathbf{m}; \sigma_s) =: D_{\Lambda_e, \boldsymbol{\lambda}_\mathbf{m}} (\mathbf{x}; \sigma_s) ,
\end{equation}
for all $\mathbf{x} \in \boldsymbol{\lambda}_\mathbf{m} + \Lambda_e $. Here the \emph{shaping variance} $\sigma_s^2$ should be taken large enough compared to $\Lambda_e$; see \cite{Viterbo-flatness}. 

\section{The AWGN channel}
\label{sec: AWGN}

In this section we briefly recall the eavesdropper's information leakage and correct decoding probability bounds in the AWGN channel, that is, when $\mathbf{h} = I_n$.  These criteria were first derived in \cite{Oggier-Sole-Belfiore, Viterbo-flatness}.  This section serves to align the analytic design criteria with our geometric intuition, and serves as a warm-up for the techniques used to analyze the fading channel.

\subsection{Case 1: Closest point decoding}

Let us study the upper bound for Eve's correct-decoding probability $P_{c, e; \Lambda_e, \Lambda_b}$, assuming that she performs a closest-point decoding on the infinite lattice $\Lambda_b$. Let $\Lambda_b, \Lambda_e \subset \R^n$ both have rank $m$. We have
\begin{equation}\label{Pce1}
P_{c, e; \Lambda_e, \Lambda_b} (\sigma)\leq
 \vol (\Lambda_b) g(\Lambda_e ; \sigma) = [\Lambda_b : \Lambda_e]^{-1} (\varepsilon_{\Lambda_e}(\sigma) + 1).
\end{equation}
This bound was first derived in \cite{Oggier-Sole-Belfiore}, and for non-full lattices in \cite{belfiore_mimo}. With the properties of the flatness factor, this implies that the probability bound is decreasing with $\sigma$ and, very intuitively, at poor signal quality $\sigma \to \infty$, the ECDP tends to the inverse codebook size $[\Lambda_b : \Lambda_e]^{-1}$, i.e  the ECDP with a ``uniform random guess''. 

As pointed out in \cite{Oggier-Sole-Belfiore}, due to the dominant term of the theta series being controlled by the sphere packing radius of $\Lambda_e$, minimizing the flatness factor can be coarsely approximated by the sphere packing problem. More rigorously but less generally, \cite{Karrila-Hollanti} proved that orthogonal lattices are always suboptimal in minimizing the theta series of a lattice. \cami{More precisely, they showed that skewing an orthogonal lattice will always reduce both the decoding error probablity for Bob and the ECDP.} Recently, \cite{Amaro_approx} provided an estimate of the flatness factor of a lattice in terms of the sphere packing radius.

We illustrate the accuracy of the above intuition with simulation results in Fig. \ref{fig:A8b}.  We have computed the ECDP bounds \eqref{Pce1} for four coset codes based on 8-dimensional lattices. The first three have $\Lambda_b=\frac{1}{2}\Z^8$, and $\Lambda_e$ has been chosen to be $\Z^8$, $L:= 2\Z \times \frac{1}{2} \Z \times \Z^6$, and the Gosset lattice $E_8$, respectively. For the fourth code, $\Lambda_e$ is the unit-volume scaling of the root lattice $A_8^*$ and $\Lambda_b = \frac{1}{2} \Lambda_e$. All these give a message set size of $[\Lambda_b : \Lambda_e] = 2^8$. Formulas for the theta functions are given in \cite{Sloane} and for $\Theta_{A_8^*}$ in \cite[Remark 2]{Chua}.  The left-hand plot in Fig. \ref{fig:A8b} shows the probability estimates obtained from the theta series as a function of the signal-to-noise ratio SNR$=10 \log_{10}(\sigma^{-2})$, and the lattices are indeed ordered by their sphere packing density. For the right-hand plot, we use $24$-dimensional lattices with $\Lambda_e = 2 \Lambda_b$, and have taken $\Lambda_e$ to be $\Z^{24}$, the Leech lattice $\Lambda_{24}$, $E_8^3$, and $E_6^4$. Again, the lattices appear in order of their sphere packing density (the minimal norms are given, \textit{e.g.}, in \cite{Sloane}). 

\begin{figure}[ht]
\begin{center}
\includegraphics[width =  0.45\textwidth]{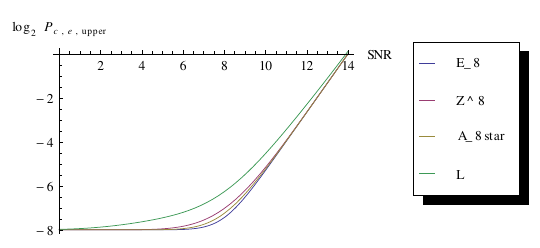}
\includegraphics[width =  0.45\textwidth]{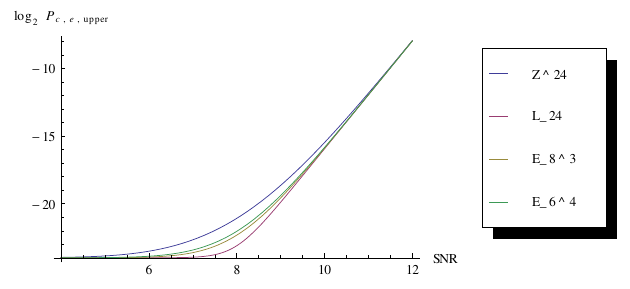}
\caption{The ECDP upper bounds \eqref{Pce1} as a function of SNR$=10 \log_{10}(\sigma^{-2})$ for certain lattices in 8 and 24 dimensions.}
\label{fig:A8b}
\end{center}
\end{figure}

\subsection{Cases 2 and 3: mod $\Lambda_s$ channel and Gaussian coset coding}
\label{subsec: AWGN setup 2}

In Case 2, we have the information leakage bound for uniform coset representatives and a mod $\Lambda_s$ channel,
\begin{equation}
\label{eq: AWGN mod lattice information bound}
I(\mathbf{m}; \mathbf{y}/\Lambda_s) \le 2 \varepsilon \log(| \mathcal{M} |) - 2 \varepsilon \log(2 \varepsilon ),
\end{equation}
where we denote and assume $\varepsilon = \varepsilon_{\Lambda_e} (\sigma) \le 1 / 2$.
This bound was proven in \cite{Viterbo-flatness}. In Case 3, with discrete Gaussian coset coding, we have the information leakage bound \cite{Viterbo-flatness}
\begin{equation}
\label{eq: AWGN discrete Gaussian information bound}
I(\mathbf{m}; \mathbf{y}) \le 8 \varepsilon \log(| \mathcal{M} |) - 8 \varepsilon \log(8 \varepsilon ),
\end{equation}
denoting and assuming $\varepsilon := \varepsilon_{\Lambda_e} (\frac{\sigma \sigma_s}{\sqrt{\sigma^2 + \sigma_s^2}}) \le 1 / 8 $. 
The information leakage bounds for fading channels derived in this paper reduce to these bounds.

\section{The Fading Channel}\label{sec: fading channel}

In this section, we generalize our channel model to an arbitrary linear fading channel model. Our aim is to motivate our choice of the various previously suggested criteria, and to provide clarity to the earlier results. The derivations are simple, hold in the general fading channel model \cami{(cf. Sec. \ref{general})}, and are valid for any real lattice coset code.  Design criteria in Rayleigh fading MIMO and SISO channels based on bounds for what the authors call \textit{average correct decision probability} and \textit{average information leakage} were derived in \cite{Belfiore-Oggier, belfiore_mimo} and \cite{mirghasemi}, with several alternative approximations. 

In Case 1, we review the derivation of the probability bounds for the general fading channel, discussing how it characterizes Eve's performance.
In Case 2, we derive a new information leakage bound for the $\bmod \ \Lambda_s$ channel. In Case 3, the average leakage and an information leakage bound from an achievability proof are known from \cite{mirghasemi,luzzi_isit16}. We derive an analytic expression for an information leakage bound in a finite-dimensional code, roughly a cross-breed of the two earlier ones. In particular, all our bounds agree on the design criterion of the average flatness factor, which we hence take as an objective function.


\subsection{Case 1}\label{subsubsec: Rayleigh setup 1}

Let $\Lambda_b, \Lambda_e \subset \R^n$ be two lattices of rank $m$, and assume that Eve simply decodes the received signal to the closest lattice point in $\Lambda_{b, \mathbf{h}}$. The probability $P_{c,e}$ of Eve correctly decoding the message is upper bounded \cite{Belfiore-Oggier, belfiore_mimo} by 
\begin{align}
%
P_{c,e}(\sigma) &\le\mathbb{E}_{\mathbf{h}}\left[ \vol (\Lambda_{b, \mathbf{h}} ) g_m(\Lambda_{e, \mathbf{h}} ; \sigma) \right]
\label{eqn:p_eve}
=\left[\Lambda_b : \Lambda_e\right]^{-1}\left(\mathbb{E}_{\mathbf{h}}\left[\varepsilon_{\Lambda_{e, \mathbf{h}}}(\sigma) \right] + 1 \right).
\end{align}
The expectation \eqref{eqn:p_eve} can be computed explicitly for several channel models. The result is of the form 
\begin{equation}
\mathbb{E}_{\mathbf{h}}\left[\varepsilon_{\Lambda_{e, \mathbf{h}}}(\sigma) \right] + 1  = \mathbb{E}_{\mathbf{h}} [ \vol (\Lambda_{b, \mathbf{h}}) g(\Lambda_{e, \mathbf{h}} ; \sigma) ] = \vol (\Lambda_b) \psi_{\Lambda_e} \left( \frac{\sigma_\mathbf{h}}{\sigma } \right)
\end{equation}
where in the last step we assume that the distribution of the fading matrix $\mathbf{h}$ is some of those explicated in Sec.~\ref{subsec: SISO model}, and $\sigma_h$ is the associated parameter.
The Rayleigh fast fading model $\psi_{\Lambda_e}$ is given by \cite{Belfiore-Oggier}:
\begin{equation}
\label{eq: def of psi series}
\psi_{\Lambda_e}\left(\frac{\sigma_\mathbf{h}}{\sigma}\right)= \psi_{\Lambda_e}^{\text{FF}} \left( \frac{\sigma_\mathbf{h}}{\sigma } \right) = \left( \frac{\sigma_\mathbf{h}}{2 \sigma } \right)^n  \sum_{\mathbf{t} \in \Lambda_e } \prod_{i = 1}^n \frac{1}{(1 + | t_i |^2 \frac{\sigma_\mathbf{h}^2}{ \sigma^2 } )^{3/2}}.
\end{equation}
\cami{The series in the above equation converges, see the Appendix for a proof. Note that (by the monotone convergence theorem) the convergence of \eqref{eq: def of psi series} ensures the finiteness of  $\mathbb{E}_{\mathbf{h}}\left[\varepsilon_{\Lambda_{e, \mathbf{h}}}(\sigma) \right]$.}

\begin{rem}
This formula admits natural generalizations to block fading channels \cite{Belfiore-Oggier} and MIMO channels \cite{belfiore_mimo}, though in what follows we will only need the explicit expression for the Rayleigh fast fading model.
\end{rem}

\begin{rem}
The scaling property of the flatness factor is inherited for all fading models,
\begin{equation}
\mathbb{E}_{\mathbf{h}}\left[\varepsilon_{\mathbf{h} \Lambda_{e}}(\sigma) \right] = \mathbb{E}_{\mathbf{h}}\left[\varepsilon_{a \mathbf{h} \Lambda_{e}}(a \sigma) \right].
\end{equation}
In particular this implies that, for a fixed fading model, the channel can be studied only by varying $\sigma$. For instance, the above explicit bounds only depend on the ratio $\sigma_\mathbf{h}/ \sigma$.
Knowing the monotonicity and limit of the flatness factor, this implies that the ECPD bound \eqref{eqn:p_eve} is for any fading channel model a decreasing function of $\sigma$, tending at poor signal quality $ \sigma \to \infty$ to $\vol (\Lambda_b) / \vol(\Lambda_e) = [\Lambda_b : \Lambda_e]^{-1}$. 
\end{rem}

\subsection{Case 2}

We derive a new information leakage bound for Case 2, a fading mod $\Lambda_s$ channel with uniform coset representatives.

\begin{thm}
\label{thm:inf_bound_mod_lambda}
In the $\bmod$ $\Lambda_s$ channel case, let $\mathbf{m}$ be a message with any distribution on the message space $\mathcal{M}$ with $| \mathcal{M} | \ge 4$. Assume that $E:= \E_{\mathbf{h}}\left[\varepsilon_{\Lambda_{e, \mathbf{h}}}(\sigma)\right] \le 1/2$. Then
\begin{align}
I\left[\mathbf{m}; (\mathbf{y}/\Lambda_{s, \mathbf{h}} , \mathbf{h})\right] &\le (1 - 2E) [ 2 E \log | \mathcal{M} | -2 E \log(2 E) ] + 2E \log | \mathcal{M} |.
\end{align}
This bound is an increasing function of $E$, attaining the trivial bound $\log | \mathcal{M} |$ at $E = 1/2$.
\end{thm}

\begin{proof}
See the Appendix.
\end{proof}

\subsection{Case 3}

We derive an information leakage bound for a fading channel with discrete Gaussian coset representatives. This setup has been considered earlier in \cite{mirghasemi,luzzi_isit16}, and our computation is a variant that yields an explicit bound as a function of the average flatness factor. We first state a lemma which gives a discrete analogue of the fact that the sum of two Gaussians is a Gaussian \cami{in the following sense: Suppose that $\mathbf{h} \in \mathbb{R}^{m \times n}$ is a deterministic matrix and $\mathbf{x} \in \mathbb{R}^n$ and $\mathbf{n} \in \mathbb{R}^m$ are independent (continuous) Gaussian vectors $\mathbf{x} \sim \mathcal{N}(\mathbf{0}, \sigma_s^2 I_n)$ and $\mathbf{n} \sim \mathcal{N}(\mathbf{0}, \sigma^2 I_m)$. Then, by basic probability theory, $\mathbf{hx} + \mathbf{y}$ is also a (continuous) Gaussian, with law $\mathcal{N}(\mathbf{0}, \sigma^2 I_{m} + \sigma_s^2 \mathbf{h}\mathbf{h}^t)$. In the lemma below we are interested in the case where $\mathbf{x}$ instead is \textit{discrete} Gaussian with parameter $\sigma_s^2$. Clearly, $\mathbf{hx} + \mathbf{y}$ is then not continuous Gaussian anymore, but the lemma quantizes the fact that (with suitable parameters) it is still fairly close to $\mathcal{N}(\mathbf{0}, \sigma^2 I_{m} + \sigma_s^2 \mathbf{h}\mathbf{h}^t)$.}

Similar estimates depicting this have been given in \cite[Lemma~5]{luzzi_isit16}, \cite[Lemma~1]{mirghasemi} and \cite[Theorem~3.1]{Peikert}.  

\begin{lemma}
\label{lem:discrete_continuous_gaussian}
Fix $\mathbf{h} \in \R^{m \times n}$ and let $\mathbf{x}$ have the centered discrete Gaussian distribution $D_{\Lambda_e,\lambda_\mathbf{m}} (\mathbf{x}; \sigma_s) $, where $\Lambda_e \subset \R^n$ is full. Let $\mathbf{n} \sim \mathcal{N}(\mathbf{0}, \sigma^2 \mathbf{I}_{m})$ be a spherical (continuous) Gaussian vector independent of $\mathbf{x}$. Assume furthermore that $\varepsilon_{\sqrt{\sigma^2/\sigma_s^2 I_{n} +  \mathbf{h}^t\mathbf{h}}\Lambda_e}(\sigma) \le \varepsilon_{\mathrm{max} }$ for some $\varepsilon_{\mathrm{max} } < 1$. Let $\tilde{\rho}(\mathbf{y})$ be the PDF of $\mathcal{N}(\mathbf{0}, \sigma^2 I_{m} + \sigma_s^2 \mathbf{h}\mathbf{h}^t)$, that is,
\begin{equation}
\tilde{\rho} (\mathbf{y}) = \frac{1}{(\sqrt{2 \pi})^m\sqrt{\det  (\sigma^2 I_{m} + \sigma_s^2 \mathbf{hh}^t)}}\exp\left(-\frac{1}{2}  \mathbf{y}^t(\sigma^2 I_{m} + \sigma_s^2 \mathbf{hh}^t)^{-1} \mathbf{y}\right) .
\end{equation}
If $\rho (\mathbf{y})$ denotes the PDF of $\mathbf{y} = \mathbf{h} \mathbf{x} + \mathbf{n}$, then
\begin{equation}
V(\tilde{\rho}(\mathbf{y}),\rho(\mathbf{y})) \leq \frac{2\varepsilon_{\sqrt{\sigma^2/\sigma_s^2 I_{n} +  \mathbf{h}^t\mathbf{h}}\Lambda_e}(\sigma)}{1-\varepsilon_{\max}}.
\end{equation}
\end{lemma}
\begin{proof}
See the Appendix.
\end{proof}

\begin{thm}
\label{thm: info bd}
Consider the fading channel with discrete Gaussian coset coding. Take a message $\mathbf{m}$ with any distribution on the message space $\mathcal{M}$ with $| \mathcal{M} | \ge 4$. Assume that 
    $E := \E_{\mathbf{h}}\left[\varepsilon_{\sqrt{ \sigma^2 / \sigma_s^2 I_{n} +  \mathbf{h}^t \mathbf{h} } \Lambda_e}(\sigma)\right] \le 1/5.$
Then,
\begin{align}
    I\left[\mathbf{m}; (\mathbf{y} , \mathbf{h})\right] \le (1-5E)[ 5 E  \log | \mathcal{M} | - 5 E \log (5 E) ] + 5 E  \log | \mathcal{M} |.
\end{align}
This bound is an increasing function of $E$, attaining the trivial bound $\log | \mathcal{M} |$ at $E = 1/5$.
\end{thm}

\begin{proof}
See the Appendix.
\end{proof}

Note that the derivations of Theorems \ref{thm:inf_bound_mod_lambda} and \ref{thm: info bd} actually hold for any $E$. For $E \ge 1/2$ or $E \ge 1/5$ the respective bounds coincide with the trivial bound $\log | \mathcal{M} |$, suggesting that the derivations are, in a way, optimal. 

 
At the first glance, the flatness factors in the information leakage bounds of Theorems \ref{thm:inf_bound_mod_lambda} and \ref{thm: info bd} are different. However, it is often reasonable to assume that the power invested in coset coding is larger than that related to the receiver's noise, i.e that $\sigma^2/\sigma_s^2 \ll 1$, and in the limit, $\sigma^2/\sigma_s^2 \to 0^+$. The following proposition shows that the two bounds coincide. A similar limit is stated in \cite{mirghasemi}. Intuitively, the more power Alice invests in coset coding, the better secrecy she has.

\begin{prop}
The expectations $E := \E_{\mathbf{h}}\left[\varepsilon_{\sqrt{ \sigma^2 / \sigma_s^2 I_{n} +  \mathbf{h}^t \mathbf{h} } \Lambda_e}(\sigma)\right]$, when they are finite, are strictly decreasing as a function of $\sigma_s$, and tend to $\E_{\mathbf{h}}\left[\varepsilon_{\mathbf{h}  \Lambda_e}(\sigma)\right] $ as $\sigma_s \to \infty$.
\end{prop}

\begin{proof}
We write the dual formula for the flatness factor \eqref{eq: dual theta flatness factor formula} in the eigenbasis of $\mathbf{h}^t \mathbf{h}$, where $\mathbf{h}^t \mathbf{h} = \diag(h_i^2)$,
\begin{align}
\label{eq: first theta series'} \varepsilon_{\sqrt{ \sigma^2 / \sigma_s^2 I_{n} +  \mathbf{h}^t \mathbf{h} } \Lambda_e}(\sigma) &= \sum_{\mathbf{t} \in \Lambda_e^*} \exp \left( -  2 \pi \sum_{i=1}^n \frac{t_i^2}{h_i^2/ \sigma^2 + 1/ \sigma_s^2}\right) - 1,
\end{align}
and
\begin{align}
\label{eq: second theta series'}
\varepsilon_{ \mathbf{h}  \Lambda_e}(\sigma) &= \varepsilon_{\sqrt{ \mathbf{h}^t \mathbf{h} } \Lambda_e}(\sigma) = \sum_{\mathbf{t} \in \Lambda_e^*} \exp \left( -  2 \pi \sum_{i=1}^n \frac{t_i^2}{1/ \sigma^2}\right) - 1.
\end{align}
The monotonicity in $\sigma_s$ is clear from \eqref{eq: first theta series'}, and the limiting property for the expected values follows by dominated convergence, since from \eqref{eq: first theta series'} and \eqref{eq: second theta series'} we have that $\varepsilon_{\sqrt{ \sigma^2 / \sigma_s^2 I_{n} +  \mathbf{h}^t \mathbf{h} } \Lambda_e}(\sigma)$ decreases to $ \varepsilon_{ \mathbf{h}  \Lambda_e}(\sigma)$ as $\sigma_s \to \infty$.
\end{proof}

\subsection{Discussion}

In both Gaussian and fading channel models, the error probability and information leakage bounds
 agree on minimizing the (average) flatness factor $\varepsilon_{  \Lambda_e}(\sigma)$ and $\E [ \varepsilon_{ \mathbf{h}  \Lambda_e}(\sigma) ]$, respectively \cami{(cf. \eqref{eq: flatness factor definition}, \eqref{eq:averageFF}).} At poor signal quality $\sigma \to \infty$, the probability bound decreases to the uniform guess probability and the information decreases to zero. The probability upper bound \eqref{Pce1} is a relatively good approximation for large $\sigma$ --- see error terms in \cite{Belfiore-Oggier}. The derivation of \eqref{eqn:p_eve} contained no new approximations after \eqref{Pce1}, and can therefore be expected to be approximative at poor signal quality.

On the contrary, the two information leakage bounds seem to be mostly suitable for achievability proofs and poor-signal asymptotics $\sigma \to \infty $. Substituting the average flatness factors $E=1/2$ and $E=1/5$ for which the respective information leakage bounds become trivial, the probability bound \eqref{eqn:p_eve} suggests a notably good secrecy. Nonetheless, the agreement of the information and probability bounds, the latter being more approximative a quantification and the former a more rigorous approach, indeed suggests that the average flatness factor should be taken as a design criterion of both practical lattice design and information-theoretic constructions.

We conclude this section by rephrasing the full wiretap problem and the design criterion: we study sublattices $\Lambda_e$ of a fixed reliability lattice $\Lambda_b \subset \R^n$, with a fixed coset code rate $\log_2 [\Lambda_b : \Lambda_e]/n$ (bits per real channel use). Equivalently, we fix the index $[ \Lambda_b : \Lambda_e]$. We design secure coset codes $\Lambda_b:\Lambda_e$ in the relevant low-SNR range by 
\begin{center}
\boxed{
\begin{minipage}{0.6\textwidth}
\begin{center}
The average flatness factor criterion: $\mathrm{minimize } \ \ \E_{\mathbf{h}} [ \varepsilon_{ \mathbf{h}  \Lambda_e}(\sigma) ].$
\end{center}
\end{minipage}
}
\end{center}

In the next section, we show how to translate this criterion into something constructive. Namely, we justify why well-rounded lattices form an appropriate family to which this minimization problem can be restricted.

\section{Well-rounded lattices and the flatness factor}
\label{sec: WR}


Finding lattices achieving local/global minima of the theta series is a long-standing problem in analytic number theory. In fact, using an assortment of tools from analytic number theory and the theory of quadratic forms, Sarnak and Str\"ombergsson \cite{Sarnak-Strombergsson} proved that, for any $\sigma > 0$, the $D_4$ lattice, the $E_8$ lattice, and the Leech lattice achieve a strict local minimum of $\Theta_\Lambda(\sigma)$ over the set of volume one lattices.
In the next section we will analyze the behaviour of $D_4$ and $E_8$ in the wiretap context and compare their performance to other lattices (in their respective dimensions).
Some of the best known sphere packings in low dimension are constructed in \cite{CostaDn} and \cite{Jorge2} from real number fields. In our simulations we will consider the construction of $D_4$ from \cite{CostaDn} and $E_8$ from \cite{Jorge2}.

Note that the lattices above maximize the packing density function (globally), so
a natural question to ask is whether or not the global minimum for the theta series function is also the densest packing in a given dimension. Unfortunately, this is not true in general. In fact, in \cite{Sarnak-Strombergsson} it was proven that the face-centered cubic lattice (fcc) being the densest packing in dimension three cannot be a global minimum for the theta series.
To make this precise, we will first show some links between the theta series and the Epstein zeta function.

The \textit{Epstein zeta function} of a lattice $\Lambda\subseteq\mathbb{R}^n$ is defined, for $s\in\mathbb{C}$ with $\Re(s) > \frac{n}{2}$, as
$$E(\Lambda,s) = \sum_{\substack{0\neq x\in \Lambda}}\frac{1}{||x||^{2s}}.$$
In analogy to the Riemann zeta function, the Epstein zeta function admits a meromorphic continuation to the complex plane with a simple pole at $s = \frac{n}{2}$.

The \textit{Gamma function} is defined as
$$\Gamma(s)=\int_{0}^{\infty}z^{s-1}e^{-z}dz,~\mbox{for}~s\in\mathbb{C},~\Re(s)>0. $$
Using an elementary change of variables we get 
\begin{equation}
\label{gamma}
 (2\pi)^{-s}||x||^{-2s}\Gamma(s)=2\int_{0}^{\infty}e^{-2\pi\sigma^2||x||^2}\sigma^{2s-1}d\sigma.  
\end{equation}
Let $\Lambda$ be a lattice in $\mathbb{R}^n$, then summing over all the non-zero lattice points in $\Lambda^{*}$ we get 
\begin{equation}
\label{epst}
     (2\pi)^{-s}\Gamma(s)E(\Lambda^*,s)=2\int_{0}^{\infty}\varepsilon_{\Lambda}(\sigma)\sigma^{2s-1}d\sigma.
\end{equation}
Equation (\ref{epst}) shows that the minima of $\Lambda\mapsto\varepsilon_{\Lambda}$ are closely related to the minima of $\Lambda\mapsto E(\Lambda,s)$. In fact, we have

\begin{equation}\label{thetaEps}
   \varepsilon_{\Lambda_1}\geq \varepsilon_{\Lambda_2}\ \forall \sigma>0 \ \Rightarrow \  E(\Lambda_1^*,s)\geq E(\Lambda_2^*,s)\  \forall s>0. 
\end{equation}

Note that all the optimization problems above are taken over $\mathcal{L}_n$ the space of volume one lattices in $\mathbb{R}^n$.

\begin{definition}
We say that a lattice $L$ is \textit{E-extremal} at $s_0\in\mathbb{C}$ if  
$$E(\Lambda,s)\geq E(L,s) ~\forall s>s_0,~\Lambda\in \mathcal{L}_n.$$
We say that $L$ is \textit{universally extremal} if
$$ E(\Lambda,s)\geq E(L,s)\ ~\forall s>0,~\Lambda\in \mathcal{L}_n.$$
\end{definition}

Notice that the universality property for $L^*$ is a necessary condition for $L$ to be a global minimum of the flatness factor $\varepsilon_{\Lambda}$.


We will still go one step further and use the connection that theta and zeta functions have with \textit{spherical designs}. In \cite{coulangeon}, Coulangeon extended the study from \cite{Sarnak-Strombergsson} and established a relation between local minima of zeta and theta functions of lattices and spherical designs.

\begin{definition}
    A \emph{spherical $t$-design} $X$ is a finite set of points on a sphere $\mathbb{S}^{n}$ such that
    $$\int_{\mathbb{S}^{n-1}}p(x)dx=\frac{1}{|X|}\sum_{x\in X}p(x),$$
    for any homogeneous polynomial $p$ of degree $\leq t$.
    
\end{definition}

There are various characterizations of a spherical design, and in what follows we will use the following alternative definition.

\begin{definition}\label{defspher}
    Let $n\geq 2$ and let $X$ be a finite subset of a sphere  $\mathbb{S}^{n}(r)$ of radius $r$. Assume that $X$ is symmetric with respect to the origin. Then, for any positive even integer $t$, $X$ is a \emph{spherical $t$-design} if and only if there exists a constant $c$ such that, for all $\alpha\in\mathbb{R}^n$,
    $$\sum_{x\in X}\langle x,\alpha\rangle^t = cr^{t/2} \langle\alpha,\alpha\rangle^{t/2},$$
    where $\langle\ ,\ \rangle$ denotes the usual scalar product.
\end{definition}


The following result connects local minimum for the Epstein zeta function and $4$-designs.

\begin{thm}[\cite{coulangeon}, Thm 4.1]\label{coul}
    Let $L$ be a lattice in $\mathcal{L}_n$ such that all its layers hold a $4$-design. Then L is E-extremal at $s$ for any $s > n/2$. If moreover $E(L, s) < 0$ for $0 < s < n/2$ , then $L$ is universally extremal.
\end{thm}

To name just a few examples of lattices with all layers holding a $4$-design, we mention $D_4$, $E_8$, the Leech lattice $\Lambda_{24},$ and Barness-Wall lattices $BW_{2m}$ of dimension $2^{2^m}$. For a complete classification up to dimension 26 see \cite{bachoc}.

Unfortunately, explicit construction of lattices  satisfying the conditions of Theorem \ref{coul}  is problematic in two ways:
\begin{enumerate}
\item Lattices having all layers holding a $4$-design do not exist in every dimension (3, 5 and 9 for example). In general dimensions there are only finitely many similarity classes of such lattices.
\item The condition $E(L, s) < 0$ for $0 < s < n/2$, implies that the Epstein zeta function has no zeros in the right half plane. Indeed, Terras \cite{terras} showed that if $\lambda_1(L)\leq \mu \sqrt{\frac{n}{2\pi e}}$, then $E(L,s)$ has a zero in $(0,n/2)$ for some $\mu<1$ and sufficiently large $n$. The bound on $\lambda_1$ is satisfied by all the known lattices in large dimension. 
\end{enumerate}
Relaxing the universality property, the minima of Epstein zeta functions at some $s_0\in\mathbb{C}$ have a tight relationship with $2$-designs in the following way.

\begin{thm}[\cite{delone}, Thm. 4]\label{zeta}
    Let $\Lambda\subset\mathbb{R}^n$ be a lattice of volume one. Then the following conditions are equivalent.
    \begin{enumerate}
        \item  There exists $s_0 > 0$ such that $\Lambda$ is a local minimum for $E(\Lambda,s)$ at $s$, for any $s > s_0$.
        \item $\Lambda$ is perfect and all layers of $\Lambda$ hold a $2$-design.
    \end{enumerate}
\end{thm}

\Tao{

 We say that $\Lambda$ is a \textit{global minimum} of the flatness factor if $\varepsilon_{\Lambda'}\geq \varepsilon_{\Lambda} $ for all $\sigma>0$ and $\Lambda'\in\mathcal{L}_n$.
We will use Theorem \ref{zeta} to prove the following proposition.
\begin{prop}
Let $\Lambda$ be a lattice in $\mathcal{L}_n$ such that $\Lambda$ is a global minimum of the flatness factor. Then the dual lattice $\Lambda^*$ is well-rounded.
\end{prop}
\begin{proof}
Let $\Lambda\in\mathcal{L}_n$ be a global minimum of the flatness factor. Then by equation \eqref{thetaEps} the lattice $\Lambda^*$ is E-extremal. In particular, there exists $s_0>0$ such that $\Lambda^*$ is a local minimum of $E(\cdot,s)$. Theorem \ref{zeta} implies that $\Lambda^*$ is perfect, hence, well-rounded.
\end{proof}
In the above proof we used the fact that a local minimum of the Epstein zeta function is perfect, and that allowed us to conclude its well-roundedness. In fact, the lattices $\Lambda$ such that $S(\Lambda)$ holds a $2$-design are also WR.}
\begin{prop}\label{spher}
    Let $\Lambda$ be a lattice such that $S(\Lambda)$ holds a $2$-design. Then $\Lambda$ is a well-rounded lattice.
\end{prop}

\begin{proof}
    Let $\Lambda\subset\mathbb{R}^n$ be a lattice such that $S(\Lambda)$ holds a $2$-design. Assume that $\Lambda$ is not WR. Then $\mathcal{S}=\mathrm{Span}_{\mathbb{R}}(S(\Lambda))$ is a proper subset of $\mathbb{R}^n$.
    Hence, there exists $\alpha\neq 0$ such that $\alpha\in \mathcal{S}^\bot$. Thus, by Definition \ref{defspher}, we get
    $$\sum_{x\in S(\Lambda)}\langle x,\alpha\rangle^2 =0=c\lambda_{1}(\Lambda)||\alpha||^2,$$
    and this contradicts the fact that $\alpha\neq 0$.
\end{proof}

\begin{cor}
   \cami{ Lattices satisfying the conditions mentioned in Theorems \ref{coul} and \ref{zeta} are  well-rounded. }
\end{cor}

\begin{rem}
As mentioned in the introduction, in \cite{luzzi} the authors proposed the maximization of the normalized product distance of the lattice and its dual in the case of the wiretap fading channel; and the maximization of the packing density of the lattice and its dual for the Gaussian wiretap channel, both of which are optimization problems that can be restricted to WR lattices. 

A good general strategy is hence to consider WR sublattices of the legitimate receiver's lattice, which is chosen to have a large non-vanishing minimum product distance (SISO) or minimum determinant (MIMO), so that also the dual of the sublattice shares these properties.

\end{rem}


\cami{
\subsection{On the construction of well-rounded algebraic lattices}



Let $K$ be a number field of degree $n=r_1+2r_2$ and $\{\sigma_1,\dots,\sigma_{r_1},\sigma_{r_1+1},\overline{\sigma_{r_1+1}},\dots,\sigma_{r_1+r_2},\overline{\sigma_{r_1+r_2}}\}$ be the embeddings of $K$ into $\mathbb{C}$. Let $\alpha$ be a totally real, totally positive element in $K$, \textit{i.e.}, $\sigma_i(\alpha)$ positive and real for any $1\leq i\leq n$. We denote by $\mathcal{A}_K$ the set of totally real, totally positive elements in $K$.

The \textit{twisted embedding} $\rho_{\alpha}~:~K\rightarrow \R^n$ is defined by
\begin{multline*}
\rho_{\alpha}(x):=(\sqrt{\sigma_1(\alpha)}\sigma_1(x),\dots,\sqrt{\sigma_{r_1}(\alpha)}\sigma_{r_1}(x),\\
\sqrt{2\sigma_{r_{1}+1}(\alpha)}\Re{\sigma_{r_{1}+1}}(x),\sqrt{2\sigma_{r_{1}+1}(\alpha)}\Im{{\sigma_{r_{1}+1}}}(x),\dots,\sqrt{2\sigma_{r_{1}+r_2}(\alpha)}\Im{{\sigma_{r_{1}+r_2}}}(x)).
\end{multline*}
If $\alpha=1$, we will simply denote the lattice $\rho_{\alpha}(\mathcal{I})$ by $\rho(\mathcal{I})$.

It is well known \cite{bayerNF} that $\rho_{\alpha}(\mathcal{I})$ is a full rank lattice in $\R^n$, where $\mathcal{I}$ is an (fractional) ideal in $\mathcal{O}_K$ the ring of integers of $K$ and $\alpha\in \mathcal{A}_K$. 
Furthermore, if $K$ is a totally real number field, then $\rho_{\alpha}(\mathcal{I})$ is a full-diversity lattice.
For a detailed exposition on algebraic number fields see for instance \cite{neukirch2013algebraic}, for lattices over number fields we refer the reader to \cite{bayerNF}.
The aim of this section is to present a generic construction of some full-diversity WR lattices that are of the form $\Lambda=\rho_{\alpha}(\mathcal{I})$. 

In \cite{fukshanskypet}, the authors studied the well-roundedness of $\rho(\mathcal{O}_K)$, \emph{i.e.},  $\mathcal{I}=\mathcal{O}_K$ and $\alpha=1$. 
\begin{thm}\cite{fukshanskypet} \label{fuksh}
Let $K$ be a number field. Then
$\rho(\mathcal{O}_K)$ is WR if and only if $K$ is a cyclotomic field, \textit{i.e.}, $K=\Q(\zeta_m)$, where $\zeta_m$ is a primitive $m^{th}$ root of unity.
\end{thm}

Theorem \ref{fuksh} ensures that $\rho(\mathcal{O}_K)$ is never WR for any real number field $K$. Thus, to construct WR lattices over real number fields $K$ (full-diversity lattices) it is necessary to consider $\rho_\alpha(\mathcal{I})$, where $\mathcal{I}$ is a strict ideal in $\mathcal{O}_K$.

In the following we will exhibit a construction of WR lattices $\rho_\alpha(\mathcal{I})$, where $\mathcal{I}$ is a (fractional) ideal in the real maximal sub-field $K$ of the $m$-th cyclotomic field $\Q(\zeta_m)$. More precisely, $K=\Q(\zeta_m+\zeta_m^{-1})$.

We recall that $\Q(\zeta_m)$ is a classic example of a complex multiplication field (CM), that is, a number field which is a totally imaginary quadratic extension of a totally real number field $K$. Our construction relies on a shifting technique proposed in \cite{bayersuarez} to obtain some of the well-known lattices as lattices over real number fields using CM fields.

Hereon, we assume that $K$ is a CM field and $F$ is the maximal totally real subfield of $K$. By definition $K$ is a quadratic extension of $F$. Let $\gamma\in F$ such that $K = F(\sqrt{\gamma})$. Assume that $-1$ is not a square in $K$, and take $K'=F(\sqrt{-\gamma})$. The field $K'$ is then a quadratic extension of $F$ different from $K$. 
The extension $K/F$ is necessarily Galois. Furthermore, $K'$ is a totally real number field.

We define $\phi~:~ K\rightarrow K'$ to be a $K$-linear map such that 
$\phi(1)=1$ and $\phi(\gamma) = \sqrt{-\gamma}$.
In \cite{bayersuarez}, the authors introduced a technique to ''shift'' lattice constructions from ideals $\mathcal{I}$ over $K$ to ideals over $K'$.
We say that an ideal class $[\mathcal{I}]$ in the class group of $K$ is \textit{ambiguous} if $\sigma([\mathcal{I}] )= [\mathcal{I}]$, where $\sigma$ is the non-trivial automorphism of the Galois group $Gal(K/F)$. Now we are ready to state the main results in \cite{bayersuarez}.
\begin{prop}
\label{Bayer}(\cite{bayersuarez})
Let $\mathcal{I}$ be an ambiguous ideal in $K$. Then $\phi(\rho_\alpha(\mathcal{I}))$ is an ideal lattice of $K'$ similar to $\rho_\alpha(\mathcal{I})$.
\end{prop}
Using the  above notation, we have the following proposition. 
\begin{prop}(\cite{bayersuarez})\label{bayersuarez}
If $K/\mathbb{Q}$ is not ramified at $2$, then $\phi(\mathcal{O}_K)$ is a fractional ideal of $K'$. Moreover, the ideal $\mathcal{I} = \phi(\mathcal{O}_K)$ satisfies $\mathcal{I}^2 = \frac{1}{2}\mathcal{O}_{K'}$.

\end{prop}
We will use Proposition \ref{bayersuarez} to construct WR lattices from totally real fields $\Q(\zeta_{4m}+\zeta_{4m}^{-1})$, where $m>1$ is an odd integer.
\begin{thm}
\label{pp}
The field $\Q(\zeta_{4m} +\zeta_{4m}^{-1})$ contains a well-rounded lattice $\rho(\mathcal{I})$ for any odd integer $m>1$. Moreover, $\mathcal{I}^2=\frac{1}{2}\mathbb{Z}[\zeta_{4m} +\zeta_{4m}^{-1}]$. 

\end{thm}
\begin{proof}
Let $m>1$ be an odd integer. The field $K=\Q(\zeta_m)$ is a quadratic extension of $F=\Q(\zeta_{m} +\zeta_{m}^{-1})$. Using the same notation as above, the corresponding field $K'$ to $K$ is $K'=\Q(\zeta_{4m} +\zeta_{4m}^{-1})$.
Using Proposition \ref{Bayer}, we get that $\rho(\phi(\mathcal{O}_K))$ is an ideal lattice over $K'$ similar to $\rho(\mathcal{O}_K)$.
Finally, $\rho(\mathcal{I})$ is WR by Theorem \ref{fuksh}, where $\mathcal{I}=\phi(\mathcal{O}_K)$.
\end{proof}

\begin{rem} Note that Theorem \ref{pp} provides explicit constructions of WR lattices in dimensions $\varphi(m)$, where $\varphi$ is the Euler totient function and $m>1$ an odd integer. The WR lattices $D_4$ and $E_8$ used in the simulations in the next section are obtained as lattices over such real maximal subfields. 
\end{rem}

}

\section{Simulations}
\label{sec: sim}
    
    All simulations in this section assume the standard Rayleigh fading channel model; fast fading in the SISO case and quasi-static in the MIMO case. We will first compare well-rounded and non-well-rounded lattices to show how the first ones outperform the second ones. \cami{After that, we will show how, when considering only WR lattices, the choice of the best one among them seems to relate to the delicate problem of simultaneously maximizing the shortest norm and minimizing the kissing number. As there is no obvious tradeoffs between these two extremes provided by the densest sphere packing lattices and GWR lattices, we have resorted to a random search aiming at finding a good sublattice with large minimal norm and small kissing number. Note that using known construction methods for WR lattices can be challenging, since we start with a fixed lattice for Bob and as well as with a fixed sublattice index determining the security level. }
    
    When using coset coding, every message (coset) is assigned several codewords.
    The sender chooses a random codeword that represents the intended message and sends it. 
    From the decoding perspective, two codewords $\lambda, \mu$ in the lattice $\Lambda_b$ represent the same message if $\lambda-\mu \in \Lambda_e$, where $\Lambda_e \subset \Lambda_b$ is a sublattice. The choice of $\Lambda_e$ can heavily influence the information leakage to a possible eavesdropper. We use the \emph{Planewalker} sphere decoder implementation \cite{Planewalker} for our coset code simulations, with around $10^6$ channel realizations per SNR. The depicted SNR is the SNR observed by Eve. However, the same plot can be used for Bob's correct decision probability just by looking at relatively higher SNRs compared to Eve's, as we assume $\sigma_e^2>\sigma_b^2$. 
    
    \cami{Note also that at the  high SNR regime, the coset code plays little role and the performance is determined by the minimum product distance. Hence, for the choice of the best sublattice, the low-to-moderate SNR regime is more interesting. In this regime, our simulations indicate that the best sublattice is a non-orthogonal WR sublattice with a large minimal norm and small kissing number. The signal constellation shape also plays a role. }

\subsection{Theta series approximation}
    In order to test what we called the average flatness factor criterion, we need to make use of the theta function of the lattices that we want to compare. Unfortunately, these theta functions consist of infinite $q$-expansion, making it impossible to compute them in an exact form. For some well-known lattices, a closed formula for the theta function is known, but in general there is no way to evaluate effectively the theta function of a given lattice. Even for the known ones, they still consist of infinite series. 
    
    In \cite{Amaro_approx} the authors propose the following approximation, referred to here as ``Barreal approximation'', which allows us to evaluate the theta function of any given lattice, allowing us in turn to compute an approximation of the expected flatness factor of the lattices we want to compare. For the flatness factor computations (Fig. \ref{FF8dim}), we have used the approximation for all lattices. 
  
    \begin{prop}[\cite{Amaro_approx}, Thm 4]
        Let $\Lambda\in\mathbb{R}^n$ be a full lattice with volume $\vol(\Lambda)$ and minimal norm $\lambda=\lambda_1^2$. The theta series $\Theta_\Lambda(q)$, where $0\leq q < 1$ can be expressed as
        $$\Theta_\Lambda(q)=(1-q^{\lambda}) - \frac{\log(q)\lambda^{\frac{n}{2}+1}\pi^{\frac{n}{2}}}{\Gamma(\frac{n}{2}+1)\vol(\Lambda)}\int_1^\infty t^{\frac{n}{2}}q^{\lambda t}dt + \mathcal{E}(\Lambda,n,L,q),$$
        where
        $$\mathcal{E}(\Lambda,n,L,q)=-C(\Lambda,n,L)\log(q)\lambda\int_1^\infty t^{\frac{n-1}{2}}q^{\lambda t}dt,$$
        and the constant $C(n,\Lambda,L)$ depends on $n$, $\Lambda$ and a Lipschiz constant $L$.
    \end{prop}
    
    We will denote by $\widetilde{\Theta}_\Lambda(q)$ the approximation $\Theta_\Lambda(q)-\mathcal{E}(\Lambda,n,L,q)$.
    We can then approximate the flatness factor as 
    $$\varepsilon_\Lambda(\sigma^2)= \frac{\vol(\Lambda)}{(\sqrt{2\pi\sigma^2})^n}\Theta_\Lambda(e^{-\frac{1}{2\sigma^2}})-1\cami{\approx} \frac{\vol(\Lambda)}{(\sqrt{2\pi\sigma^2})^n} \left(\left(1-e^{-\frac{\lambda}{2\sigma^2}}\right) + \frac{(\lambda\pi)^{\frac{n}{2}}}{2\sigma^2\Gamma(\frac{n}{2}+1)\vol(\Lambda)} \int_1^\infty t^{\frac{n}{2}}e^{-\frac{\lambda t}{2\sigma^2}}dt\right).$$
    
   In Figure \ref{fig:thetas} we can see how good this approximation is for some lattices, comparing it with different truncations of their theta series $q$-expansions. For a more rigorous study, we refer to \cite{Amaro_approx}.
    
    \begin{figure}[ht]
	    \begin{center}
		\includegraphics[scale=0.28]{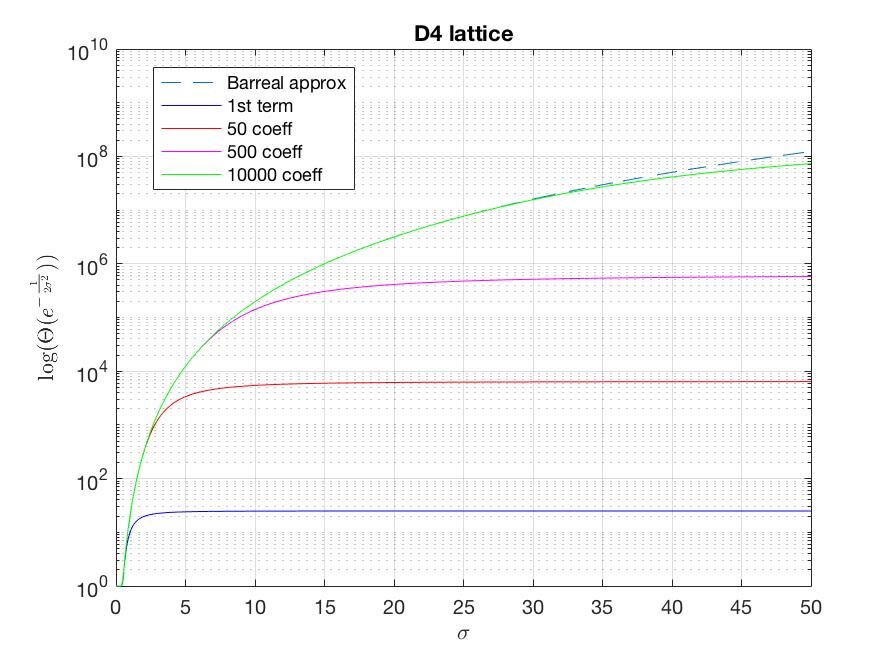}
		\includegraphics[scale=0.28]{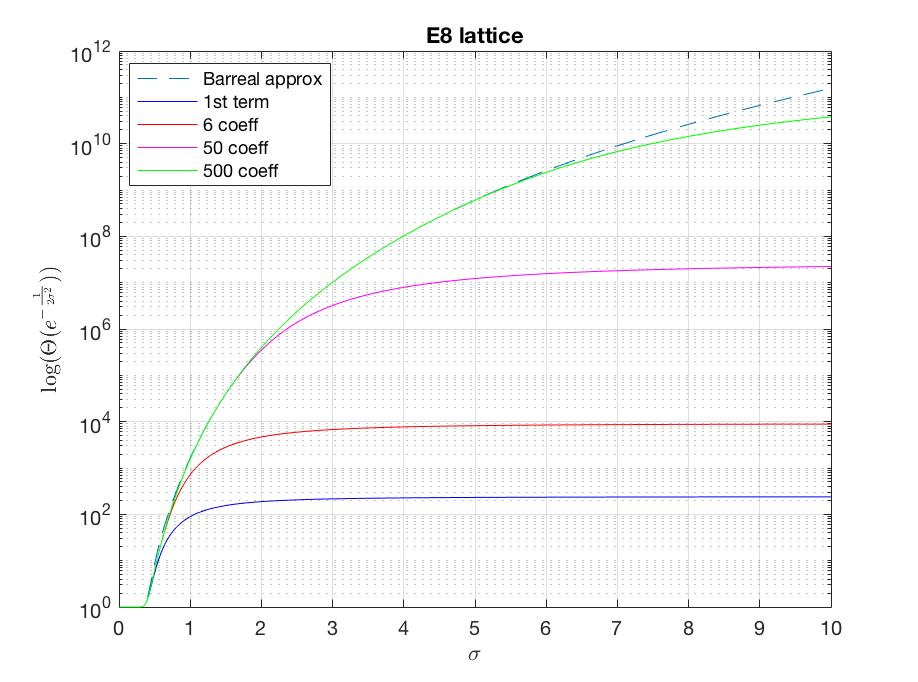}
		\includegraphics[scale=0.3]{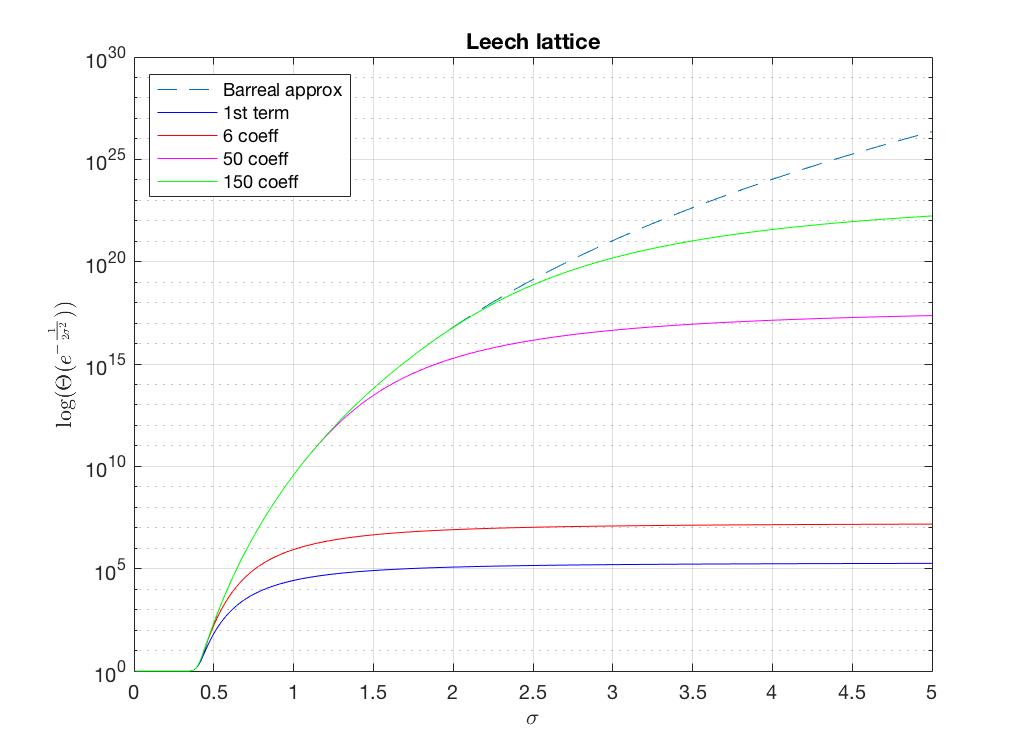}
		\caption{Comparison of the Barreal \emph{et al.} approximation \cite{Amaro_approx}     $\widetilde{\Theta}_{\Lambda}(e^{-\frac{1}{2\sigma^2}})$ with different truncated $q$-expansions of $\Theta_{\Lambda}(e^{-\frac{1}{2\sigma^2}})$.}
		\label{fig:thetas}
	\end{center}
	\end{figure}
     
\subsection{Well-rounded vs non-well-rounded lattices}
     
    For WR lattices, we first restrict the code lattices $\Lambda_b$ to orthonormal lattices, including: the Alamouti code, the Golden code and some full diversity algebraic rotations as given in the table \cite{Viterbo_tables}. These lattices yield good performance in the high SNR regime, their orthogonal structure simplifies calculations, and they are therefore realistic candidates for actual implementations. We found well-rounded sublattices $\Lambda_e$ of $\Z^n$ with prescribed index, and hence of the aforementioned orthonormal codes, by using a probabilistic search algorithm.  We compared these against scalings of $\Z^n$ and a non-well-rounded lattice in a fast-fading Rayleigh SISO channel as well as in quasi-static MIMO channels.

\subsubsection{SISO}
    For the SISO channel we simulate a coset coding protocol for a wiretap channel with the three following choices for $\Lambda_e$:
    {\begin{center} $
	    \Lambda_1=\begin{pmatrix}
	    16 &0 &0 &0\\
	    0&4&0&0\\
	    0&0& 2&0\\
	    0&0&0& 2
	    \end{pmatrix} \Z^4, \,\,
	    \Lambda_2=\begin{pmatrix}
	    4 &0 &0 &0\\
	    0&4&0&0\\
	    0&0& 4&0\\
	    0&0&0&4
	    \end{pmatrix} \Z^4, \,\,
	    \Lambda_3=\begin{pmatrix}
	    -2 &-3 &4 &-1\\
	    0&-1&0&3\\
	    0&-3& -2&-3\\
	    -4&-1&0&-1
	    \end{pmatrix} \Z^4
	    $
    \end{center}} \medskip
    
    Their parameters are compiled in Table \ref{tab:Z4params}. All sublattices have the same index, so we have the same data rate in all three cases. The non-well-rounded lattice has the lowest minimal norm, while $\Lambda_3$ almost achieves the theoretical maximum of $\gamma_4 256^{\frac{2}{4}}=16\sqrt{2}\approx 22.63$, where $\gamma_n$ is the $n$-th Hermite constant. The optimal packing in 4 dimensions is given by the $D_4$ lattice, but although it is a sublattice of $\Z^4$ no integer scaling of $D_4$ has the desired index $256$. 
    
     \begin{table}[ht]
	    \centering
	    \begin{tabular}{ccccccccc}
		    \toprule
		    Lattice	& $\lambda_1(\Lambda)^2$ & $\kappa$ & WR & index & $R_i$ & $R_c$ 
		    & WR dual \\
	    	\midrule
		    $\Lambda_1$ & 4  & 4 & no & 256 & 2 & 2 & no \\ 
		    $\Lambda_2$ & 16  & 8 & yes & 256 & 2 & 2 & yes \\ 
		    $\Lambda_3$ & 20  & 12 & yes & 256 & 2 & 2  & yes \\ 
	    \end{tabular}
	    \caption{Parameters for three different sublattices of $\Z^4$ with kissing number $\kappa$.}
	    \label{tab:Z4params}
    \end{table}
    
    The simulation results for the ECDP are depicted in Figure \ref{fig:SimZ4}. We see that the well-rounded lattices clearly outperform the non-well-rounded sublattice and $\Lambda_3$ approaches the theoretical minimum of $\frac{1}{256}$, which would be equivalent to Eve guessing uniformly at random among the 256 different messages.

    \begin{figure}[ht]
	    \begin{center}
		\includegraphics[scale=0.6]{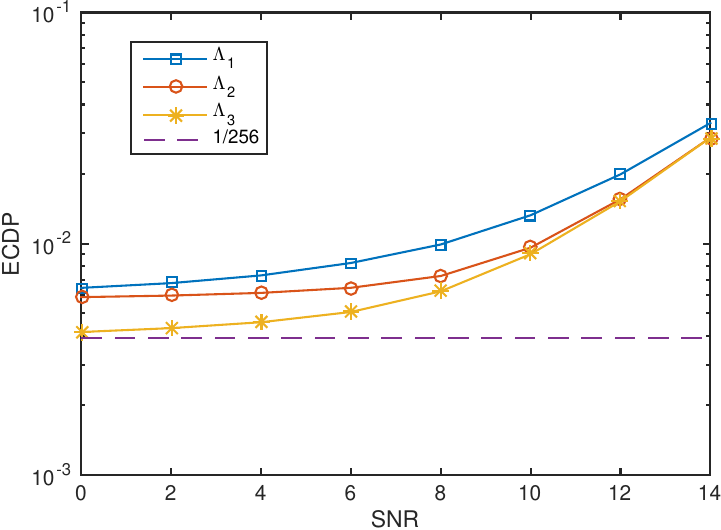}
		\caption{Simulation of ECDP for three different $\Z^4$ sublattices.}
		\label{fig:SimZ4}
	\end{center}
	\end{figure}

\subsubsection{MIMO}
    The Alamouti code is defined as a sublattice of $\C^4$ of rank 4. Since it is orthogonal it can be seen as an isomorphic image of $\Z^4$ and we can use the same sublattices $\Lambda_2$ and $\Lambda_3$. We take $\Lambda_4 := 2\Lambda_2$ and $\Lambda_5 := 2\Lambda_3$. Both lattices are well-rounded but the non-orthogonal choice, with longer minimal vectors, clearly outperforms the trivial choice of a scalar multiple of $\Z^4$ in the low SNR regime, as seen in Fig. \ref{fig:Alamouti-256}. Both sublattices behave identically for high SNR since the coset coding rarely comes into play in these situations.

    \begin{figure}[ht]
	    \centering
	    \includegraphics[scale=0.5]{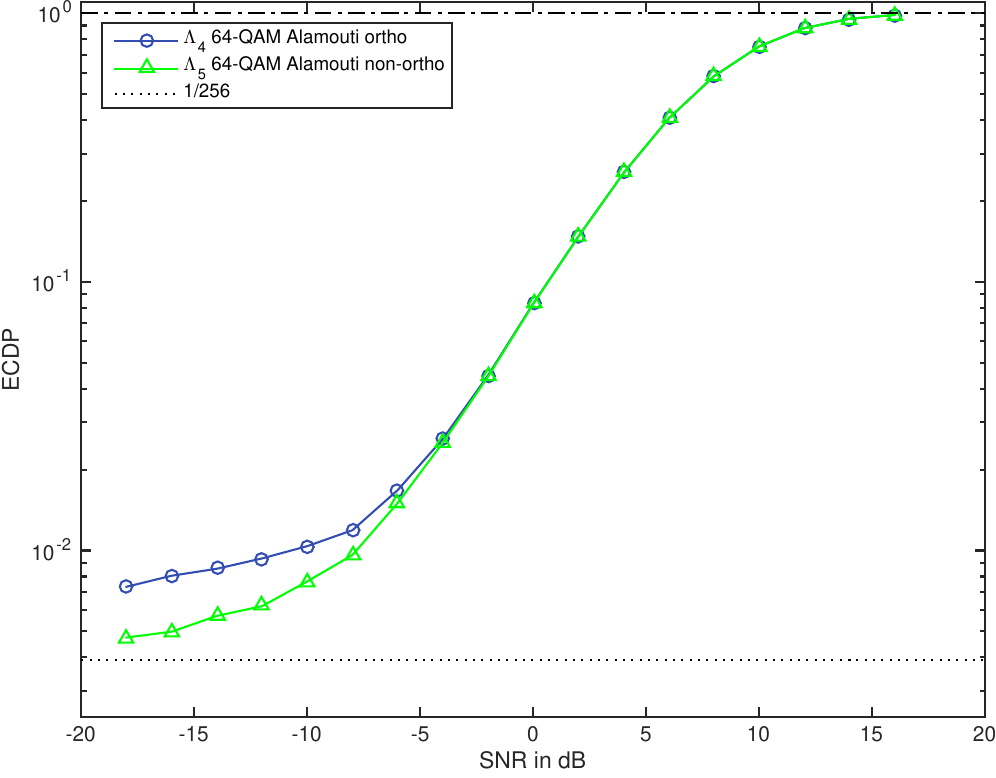}
    	\caption{ECDP for 64-QAM Alamouti code.}
    	\label{fig:Alamouti-256}
    \end{figure}

    For the Golden code we choose sublattices $\Lambda'_1, \Lambda'_2, \Lambda'_3$ of $\Z^8$ with the following respective generator matrices.
    \begin{align*} 
    \Lambda'_1 &= 2\cdot \diag(2,2,2,2,2,1,1,1), &
    \Lambda'_2 &= 2\cdot\begin{bsmallmatrix}
    1& 0& 1& 0& 0& 0& 1& 1\\
    0& 0& 0& 0& 1& -1& 0& -1\\
    1& 1& 0& 0& 0& -1& -1& 0\\
    0& -1& -1& 1& 0& 0& 0& 0\\
    0& 0& 0& 1& 0& 1& 0& -1\\
    -1& 0& 0& 1& 0& 0& 0& 0\\
    0& 1& -1& 0& 1& 0& 0& 0\\
    0& 0& 0& 0& -1& 0& 1& 0\\
    \end{bsmallmatrix},  &
    \Lambda'_3 &= 2\cdot \begin{bsmallmatrix} 
    -1& -1& -1& 1& 0& 0& 0& 0\\
    1& 0& 0& 1& 0& 0& 0& 0\\
    0& 0& 1& 1& 0& 0& 0& 0\\
    0& -1& 0& 0& -2& 0& 0& 0\\
    -1& -1& 0& 0& 0& 2& 0& 0\\
    -1& 0& 1& 0& 0& 0& 0& 2\\
    0& 1& 0& -1& 0& 0& 0& 0\\
    0& 0& -1& 0& 0& 0& -2& 0\\
    \end{bsmallmatrix}.
    \end{align*}

    All three lattices have index $32$ so we can encode $5$ bits in each message, or equivalently achieve a rate of $2.5$ bits per channel use. We vary the codebook size to allow for different rates of confusion $r_c$ as detailed in Table \ref{tab:sublatGold}.

			
    \begin{table}[ht]
	\setlength{\tabcolsep}{4pt}
	\begin{center}
		\begin{tabular}{c c c c c  | r r r | r r r | r r r}
			\toprule
			& & & & &
			\multicolumn{3}{|c|}{4-QAM} & \multicolumn{3}{c|}{16-QAM} & \multicolumn{3}{c}{64-QAM}\\
			& index	& $\lambda_1(\Lambda)^2$ & WR  & WR dual & $r$ & $r_i$ & $r_c$ & $r$ & $r_i$ & $r_c$ & $r$ & $r_i$ & $r_c$ \\ 
			\midrule
			$\Lambda'_1$ & 32  & 4 & no &  no & & & & & & & \\
			$\Lambda'_2$ & 32 & 12 & yes &  yes &  4 & 2.5 & 1.5 & 8 & 2.5 & 5.5 & 12 & 2.5 & 9.5 \\
			$\Lambda'_3$ & 32 & 16 & yes &  yes & & & & & & & \\
		\end{tabular}
		
		\caption{Sublattices of the Golden Code of index $32$}	\label{tab:sublatGold}
	\end{center}
    \end{table}
    
    The simulation results in Figure \ref{fig:Golden} clearly show that increasing the rate of confusion has strong implications for the ECDP. Increasing $r_c$ from $1.5$ bpcu to $5.5$ bpcu leads to an SNR shift of about 10 dB and increasing $r_c$ to $9.5$ bpcu increases the gap by another 8 dB. Both well-rounded lattices $\Lambda'_2$ and $\Lambda'_3$ perform very similarly despite their different minimal lengths.

    \begin{figure}[ht]
	    \centering
	    \includegraphics[scale=0.5]{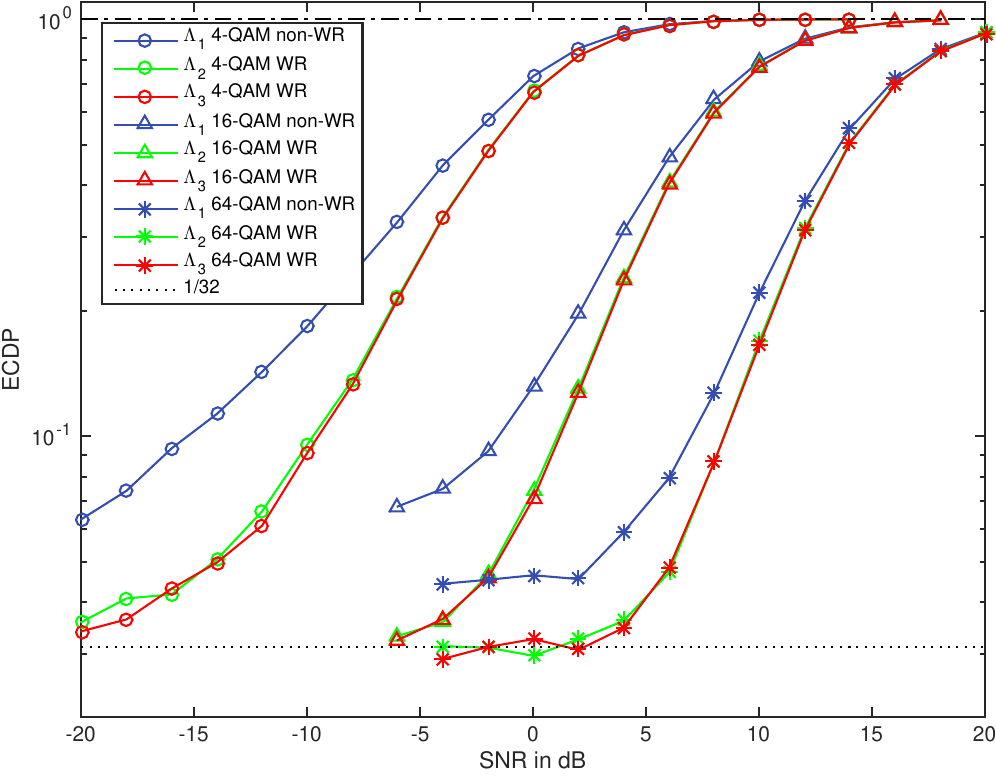}
	    \hspace{-1.0pt}
	    \caption{ECDP for 4-QAM, 16-QAM and 64-QAM Golden code.}
	    \label{fig:Golden}
    \end{figure}

\subsection{Choosing among well-rounded lattices}
    Once we have seen that well-rounded lattices perform better than non-well-rounded ones, and that among WR lattices, the non-orthogonal ones are better than the orthogonal ones, we are interested in choosing among well-rounded lattices. \cami{To choose the best one among WR lattices, it seems beneficial to have a large minimal norm and small kissing number.} We will illustrate this in some examples for dimensions $n=4$ and $n=8$.

\subsubsection{$n=4$}
    We consider WR sublattices of the optimal rotation of $\mathbb{Z}^4$ in terms of the minimum product distance proposed by Viterbo \emph{et al.} \cite{viterbo2,Viterbo_tables}. We refer to this rotated $\Z^4$ lattice by \textit{Krus4}.
    We will compare the sublattice having the best sphere packing shape to a randomly found non-orthogonal well-rounded sublattice of Krus4 (denoted by \textit{WR nonortho} in the table below), which is not isometric to the best sphere packing. 
    The basis matrix is given by
    $$\left(\begin{smallmatrix}2.6179140488&-2.414499458&-2.6174487992&0.6824076544\\
        -3.0009700976&-2.338015074&2.3164903194&0.4021649334\\
        3.4332588812&-1.6385571204&1.3956732198&-1.8920783058\\
        0.8223162707&-1.4725414834&3.5158030921&2.1896452076\\
    \end{smallmatrix}\right).$$
    For the generators of the rotation, see \cite{Viterbo_tables}.
    The choice of this sublattice is based on trying to simultaneously obtain a long shortest vector and a low kissing number, along the lines suggested previously.
    
    In addition to the rotated $\Z^4$ lattice, we will consider the algebraic construction of $D_4$ described in \cite{spawc}.
    In Table \ref{4dim} we depict the minimum product distance, shortest norm, kissing number, and the average energy for the lattices under comparison. We can see how the PAM signaling causes suboptimal average energy for the $D_4$ superlattices due to its skewness. It is also evident from the following simulations that the relationship between the shortest vector length and kissing number is an interesting topic for further study.

            
    \begin{table}[ht]
        \centering
        \vspace{0.15cm}
        \begin{tabular}{|c|l|c|c|c|c|}
            \hline
            Lattice   & $d_{p,\min}^{1/n}$ & $\lambda_1(\Lambda)^2$ &  $\kappa$ & $E_{av}$ \\ \hline \hline
            \textbf{Krus4} (Viterbo \cite{Viterbo_tables}) & 0.43899 & 1 & 8 & 84 \\ \hline
            4Krus4 & & 16  & 8 & \\ \hline
            WR nonortho & & 20  & 12 &  \\ \hline\hline
            $\mathbf{D_4}$ (Costa \cite{CostaDn}) & 0.33856 & 1.41421 &  24 & 118  \\ \hline
            $\mathrm{4D_4}$ & & 22.62735  & 24 & \\ \hline
        \end{tabular}
        
        \hspace{1cm}
        \caption{Viterbo et al. $\Z^4$ algebraic rotation vs Costa et al. $D_4$ algebraic rotation. Superlattices \textbf{boldfaced} and normalized to unit volume. $E_{av}$ stands for the average energy of the overall codebook with $8$-PAM.}
        \label{4dim}
    \end{table}
    
    Fig. \ref{4dim16pam} shows the performance of different 4-dimensional lattices using 16-PAM. The results with 8-PAM are very similar, except that the crossing point of the curves is slightly above $\ECDP=10^{-2}$. 
    \begin{figure}[ht]
        \centering
        \includegraphics[width=9cm]{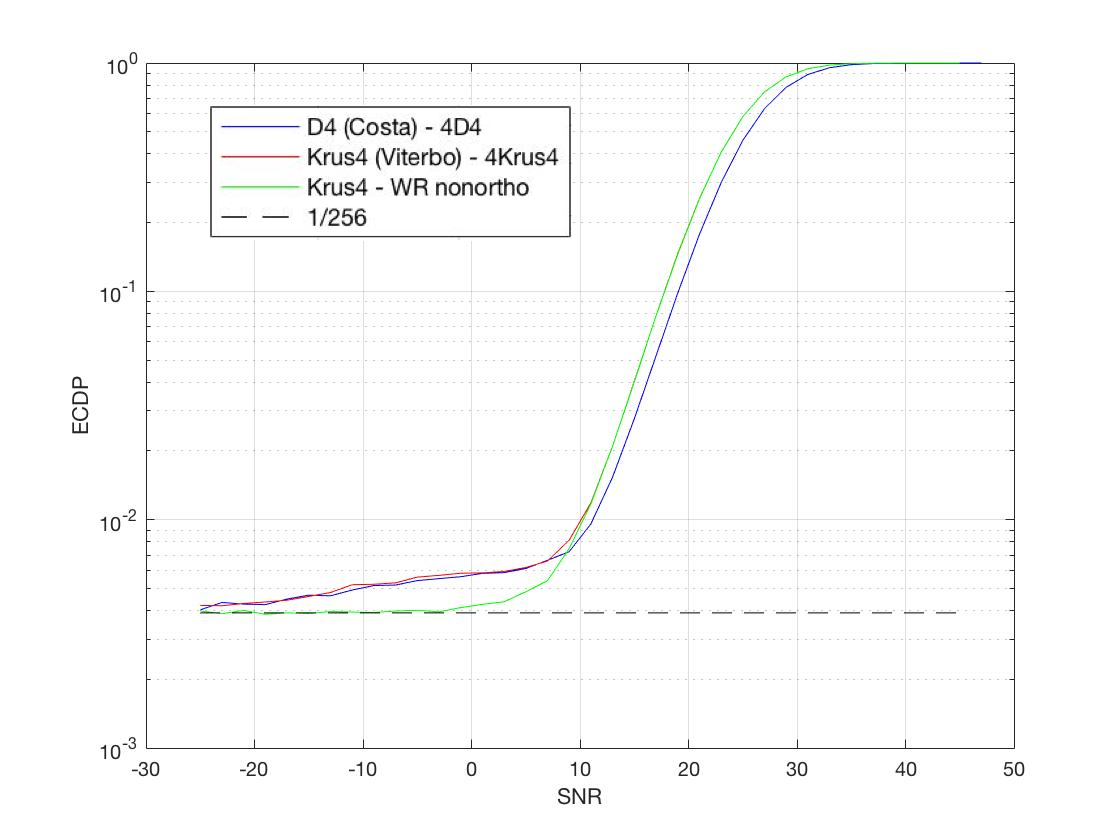}
        \caption{Comparison of 4-dimensional WR lattices with 16-PAM.  Sublattice index 256.}
        \label{4dim16pam}
    \end{figure}

\subsubsection{$n=8$}
    Analogously to the dimension 4 case, we consider WR sublattices of the optimal rotations of $\mathbb{Z}^8$ in terms of the minimum product distance proposed in \cite{viterbo2,Viterbo_tables}. We refer to this rotated $\Z^8$ lattice by \textit{Cyclo8}.
    The \textit{Cyclo8} lattice has a rotated $2E_8$ as a sublattice since $2E_8\subseteq \mathbb{Z}^8$. We will compare the sublattice having the best sphere packing shape to randomly found non-orthogonal well-rounded sublattices of \textit{Cyclo8} (denoted by \textit{WR nonortho} in the table below), which is not isometric to the best sphere packing. 
    The basis matrix is given by
    $$\left(\begin{smallmatrix}
        0.7007035387&-1.8196532242&0.5098308731&0.5436478577& -0.354502581&1.0846703755&-0.47456554011&-0.3391191908 \\
        -0.1711086948&-1.2014038879&0.2876791984&-0.8163368876&0.70784566949&-0.3945408211&-0.8973956636&1.5218930146 \\
        0.40388427868&-0.3992399381&0.359383021&-0.3556171703& -0.8543195887&-0.4910135947&-1.6941162987&-1.2573300636 \\
        0.5710190406&1.817991942&0.7541367758&0.607964929 & 0.1228109789&-0.9716075716&-0.13023119099&-0.6741179039 \\
        -0.84701159571&-1.0127325492&0.8324416454&0.0586435438 &-0.09880446362&-1.7873172339&-0.5616527529&0.2020567981 \\
        -0.04341452009&-0.3892212211&1.5903470277&-0.8144416018 & 0.8859851713&-0.084615344&1.3643633956&-0.022051583\\
        1.70075388399&-0.6502044959&0.3252620337&-1.1219451349 & -0.89655915559&-0.3737755366&0.10194674011&-0.6051400279\\
        -1.5754499665&0.8921488186&0.5880215656&-0.3340078386 &  -1.0081995324&1.0730373751&0.0358312991&-0.3090920604\\
    \end{smallmatrix}\right)$$

    We will also consider two different constructions of $E_8$ described in \cite{spawc}.
    In Table \ref{8dim} we depict the minimum product distance, shortest norm, kissing number, and the average energy for the lattices under comparison. We can see how the PAM signaling causes suboptimal average energy for $E_8$ due to its skewness. As before, one can see from the following simulations that the relationship between the shortest vector length and kissing number is nontrivial.

      
    \begin{table}[ht]
        \centering
        \begin{tabular}{|c|c|c|c|c|c|}
            \hline
            Lattice  & $d_{p,\min}^{1/n}$ & $\lambda_1(\Lambda)^2$ &  $\kappa$ & $E_{av}$ \\ \hline\hline
            \textbf{Cyclo8} (Viterbo \cite{viterbo2}) & 0.289520 & 1  & $16$ & $40$ \\ \hline
            2Cyclo8 & & 4  & 16 & \\ \hline
            WR nonortho & &  6  & 40 &   \\ \hline
            $\mathrm{2E_8}$ & & 8 &  240 & \\ \hline \hline
            $\mathbf{E_8}$ (Costa \cite{Jorge2}) & 0.293826  & 2 &  240 & 80 \\ \hline
            $\mathrm{2E_8}$ & & 8 &  240 & \\ \hline
        \end{tabular}
        
        \hspace{1cm}
        \caption{Viterbo et al. $\Z^8$ algebraic rotation  vs Costa et al. $E_8$ algebraic rotation. Superlattices \textbf{boldfaced} and normalized to unit volume.  $E_{av}$ stands for the average energy of the overall codebook with $4$-PAM.}
        \label{8dim}
    \end{table}

    Fig. \ref{FF8dim} displays the ECDP upper bounds in terms of the expected flatness factors (EFF) as explained in the previous section, computed by using the approximation from \cite{Amaro_approx}. We can see that for 8-dimensional WR lattices (behavior of 4-dimensional lattices is very similar) the curves almost coincide, and the small differences do not necessarily give the performance order observed in the simulations that follow. Hence, choosing the  lattice with the smallest EFF alone may not yield the desired outcome. Furthermore, even though the differences in the EFFs of the lattices are negligible, the simulations show a very clear difference in the actual ECDP performance, even beyond 10dB depending on the SNR range.
    
    \begin{figure}[ht]
        \centering
        \includegraphics[width=11cm]{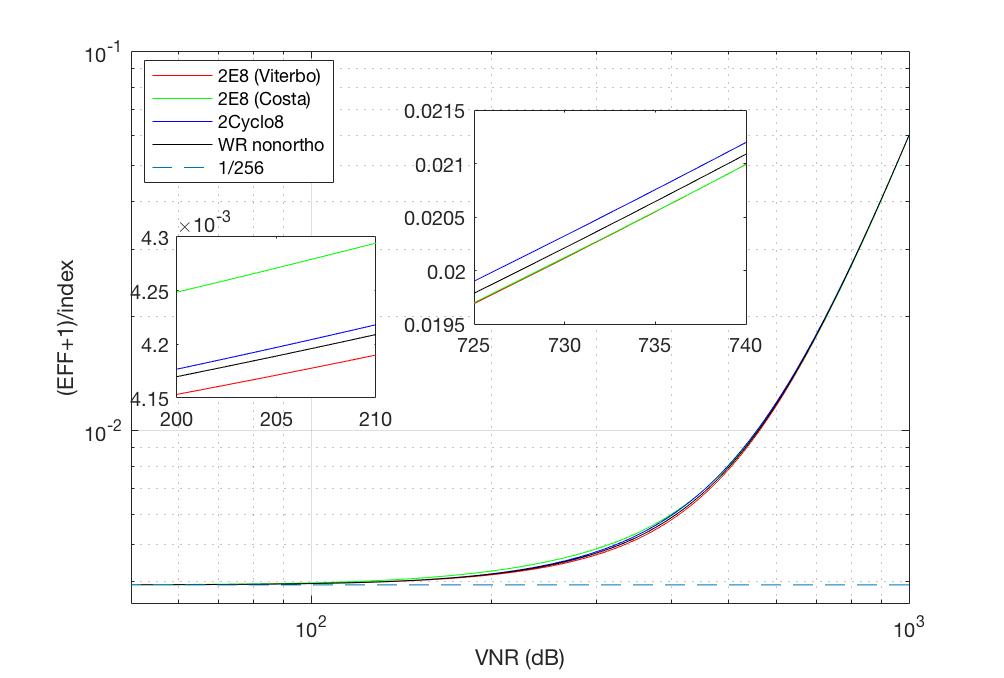}
        \caption{Comparison of the ECDP upper bounds $(\mathrm{EFF}+1)/\mathrm{index}$ for 8-dimensional WR sublattices with respect to the volume-to-noise ratio (VNR).}
        \label{FF8dim}
    \end{figure}
     \begin{figure}[ht]
        \centering
        \includegraphics[width=12cm]{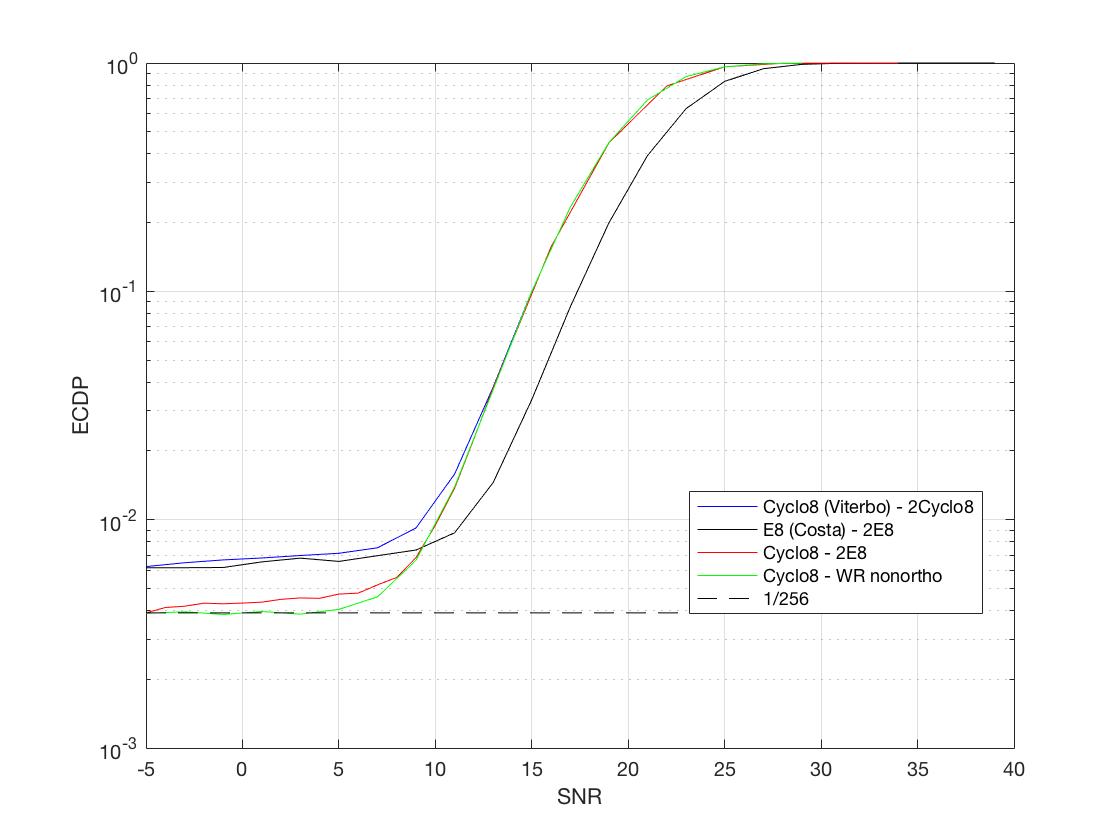}
        \caption{Comparison of 8-dimensional WR lattices with 8-PAM. Sublattice index 256.}
        \label{8dim8pam}
    \end{figure}
    
    In Fig. \ref{8dim8pam} different 8-dimensional lattices are compared with 8-PAM. Again, the results with 4-PAM are very similar, except that all the curves ultimately get to the asymptote, and the crossing point is slightly higher, just above $\ECDP=10^{-2}$.
    Both in dimension 4 and 8, the well-rounded non-orthogonal lattice outperforms the other lattices in the regime close to the asymptote corresponding to uniform guessing. The curves cross at higher SNRs, after which the Costa $D_4, E_8$ constructions become better,  so the choice of the best lattice also  depends on the target ECDP and Bob's expected SNR. \cami{Note, however, that at higher SNRs lowering Eve's decision probability may also worsen Bob's. Ideally, Bob's and Eve's SNRs would be located at the different ends of the ``waterfall'' region. In all the simulations, if this is the case then the best choice is uniquely the same from both ECDP and Bob's error probability point of view: the one minimizing ECDP at low SNR is the same as the one minimizing Bob's error probability at high SNR.  }

     We also remark again that using PAM signaling for the non-orthogonal superlattices is suboptimal (more so than for orthogonal lattices) in terms of average transmission power.  To this end, we also ran simulations with spherical shaping, depicted in Fig. \ref{8dim8pam_sphere}. The differences get a lot smaller, and the crossing point seems to disappear. In the interesting low SNR regime, the WR nonorthogonal sublattice still outperforms the others. Furthermore, spherical encoding and decoding are much more complex than that of cubically shaped (\emph{i.e.}, with a symmetric PAM coefficient set) constellations.  It also makes bit labeling a delicate problem.

    \begin{figure}[ht]
        \centering
        \includegraphics[width=12cm]{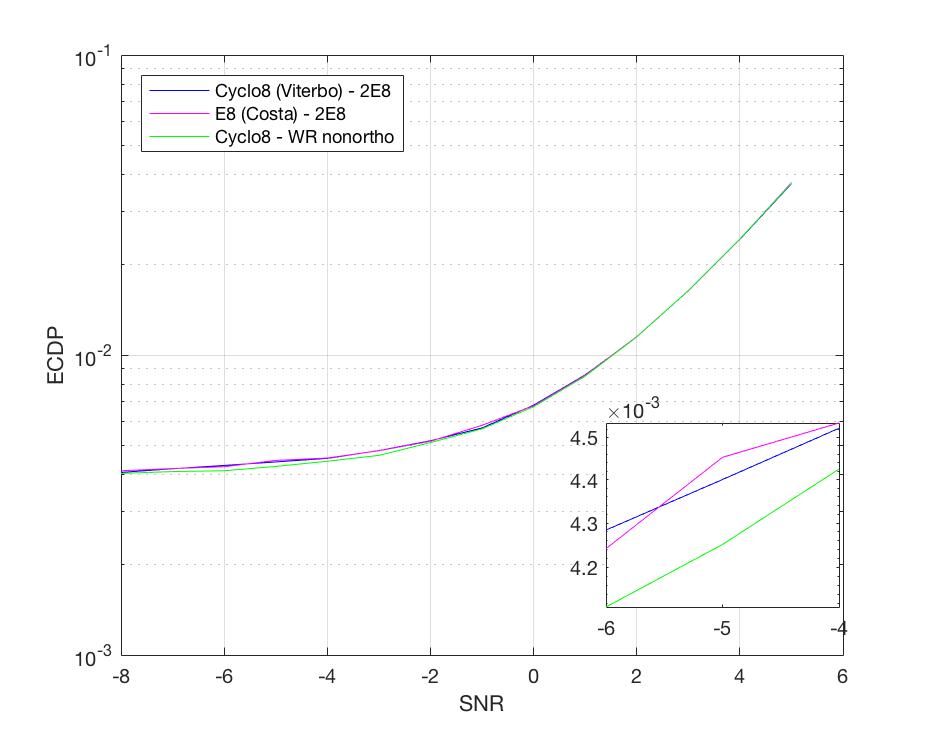}
        \caption{Comparison of 8-dimensional WR lattices with spherical shaping; $2^{16}$ shortest vectors  chosen out  of the $2^{24}$ vectors gotten from 8-PAM. Sublattice index 256.}
        \label{8dim8pam_sphere}
    \end{figure}

\section{Conclusions and future work}
In this paper, we studied how to constructively design lattices for a fixed dimension with the goal of minimizing the expected lattice flatness factor, which is crucial for minimizing both the eavesdropper's correct decision probability and leaked information. While the flatness factor is closely related to the lattice theta series, which is very difficult to analyze, we showed that the optimization problem can be restricted to the family of well-rounded lattices. Well-rounded lattices can be constructed by using, \emph{e.g.}, algebraic methods or random search, and hence provide an efficient means to produce well-performing wiretap codes. It was also pointed out that some previous, asymptotic design criteria can also be met by WR lattices. Namely, the minimum product distance and minimum determinant maximization problems can be restricted to WR lattices without loss of generality. Our extensive simulations in SISO and MIMO channels show that WR lattices outperform non-WR ones, and moreover demonstrate that the flatness factor is not sufficient alone for choosing the best lattice performance-wise.  \cami{Hence, in conclusion, a good general strategy combining the results of the present and previous works is to choose the eavesdropper's sublattice to be non-orthogonal and WR with large minimal norm, small kissing number, and with large non-vanishing minimum product distance/minimum determinant for both the lattice and its dual. }

As future work, it is still open how to optimally choose among the WR lattices in order to achieve the best performance. Deriving explicit constructions for different dimensions will aid studying this question, as well as provide a repository for good wiretap lattices. \cami{The notion of generic well-rounded lattices seems also interesting as a smaller kissing number seems beneficial at least for low SNRs. Optimizing the shortest norm among GWR lattices may turn out to provide a solution to the first problem. Another interesting and most likely hard  question is to characterize a tradeoff between the shortest norm and the kissing number among WR lattices.}

\section*{Acknowledgements}

D.\ Karpuk was supported by Academy of Finland grant \#268364. C.\ Hollanti was supported by Academy of Finland grants \#276031, \#282938,  \#283262, and \#303819, and by the Technical University of Munich -- Institute for Advanced Study, funded by the German Excellence Initiative and the EU 7th Framework Programme under grant agreement  \#291763, via a Hans Fischer Fellowship. The support from the Finnish Cultural Foundation to O.~Gnilke and from the European Science Foundation under the ESF COST Action IC1104 to the ANTA Group are
also gratefully acknowledged. We thank L. Luzzi, R. Vehkalahti, and C. Ling for bringing their  results \cite{luzzi_isit16} to our attention.



\begin{thebibliography}{10}
\providecommand{\url}[1]{#1}
\csname url@samestyle\endcsname
\providecommand{\newblock}{\relax}
\providecommand{\bibinfo}[2]{#2}
\providecommand{\BIBentrySTDinterwordspacing}{\spaceskip=0pt\relax}
\providecommand{\BIBentryALTinterwordstretchfactor}{4}
\providecommand{\BIBentryALTinterwordspacing}{\spaceskip=\fontdimen2\font plus
\BIBentryALTinterwordstretchfactor\fontdimen3\font minus
  \fontdimen4\font\relax}
\providecommand{\BIBforeignlanguage}[2]{{%
\expandafter\ifx\csname l@#1\endcsname\relax
\typeout{** WARNING: IEEEtran.bst: No hyphenation pattern has been}%
\typeout{** loaded for the language `#1'. Using the pattern for}%
\typeout{** the default language instead.}%
\else
\language=\csname l@#1\endcsname
\fi
#2}}
\providecommand{\BIBdecl}{\relax}
\BIBdecl

\bibitem{Gnilke}
O.~W. Gnilke, H.~T.~N. Tran, A.~Karrila, and C.~Hollanti, ``Well-rounded
  lattices for reliability and security in {R}ayleigh fading {SISO} channels,''
  in \emph{Proc. {IEEE} Inf. Theory Workshop}, 2016, pp. 359--363.

\bibitem{Barreal-Karrila-Karpuk-Hollanti}
A.~Barreal, A.~Karrila, D.~A. Karpuk, and C.~Hollanti, ``Information bounds and
  flatness factor approximation for fading wiretap {MIMO} channels,'' in
  \emph{Proc. {IEEE} Int. Telecommunication Netw. Appl. Conf.}, 2016, pp.
  277--282.

\bibitem{Gnilke-Barreal}
O.~W. Gnilke, A.~Barreal, A.~Karrila, H.~T.~N. Tran, D.~A. Karpuk, and
  C.~Hollanti, ``Well-rounded lattices for coset coding in {MIMO} wiretap
  channels,'' in \emph{Proc. {IEEE} Int. Telecommunication Netw. Appl. Conf.},
  2016, pp. 289--294.

\bibitem{spawc}
M.~T. Damir, O.~Gnilke, L.~Amor{\'o}s, and C.~Hollanti, ``Analysis of some
  well-rounded lattices in wiretap channels,'' in \emph{Proc. IEEE Int.
  Workshop on Signal Process. Adv. in Wireless Commun.}, 2018, pp. 1--5.

\bibitem{wiretap-old}
A.~Karrila, D.~Karpuk, and C.~Hollanti, ``On analytical and geometric lattice
  design criteria for wiretap coset codes,'' \emph{arXiv preprint
  arXiv:1609.07723v2}, 2016.

\bibitem{Wyner}
A.~D. Wyner, ``The wiretap channel,'' \emph{Bell system technical journal},
  vol.~54, no.~8, pp. 1355--1387, 1975.

\bibitem{Wyner-Ozarow}
L.~H. Ozarow and A.~D. Wyner, ``Wire-tap channel {II},'' \emph{AT\&T Bell
  Laboratories technical journal}, vol.~63, no.~10, pp. 2135--2157, 1984.

\bibitem{Oggier-Sole-Belfiore}
F.~{Oggier}, P.~{Sol\'e}, and J.~{Belfiore}, ``Lattice codes for the wiretap
  {G}aussian channel: Construction and analysis,'' \emph{IEEE Transactions on
  Information Theory}, vol.~62, no.~10, pp. 5690--5708, 2016.

\bibitem{Viterbo-flatness}
C.~Ling, L.~Luzzi, J.-C. Belfiore, and D.~Stehl{\'e}, ``Semantically secure
  lattice codes for the {G}aussian wiretap channel,'' \emph{{IEEE} Trans. Inf.
  Theory}, vol.~60, no.~10, pp. 6399--6416, 2014.

\bibitem{luzzi}
L.~Luzzi, R.~Vehkalahti, and C.~Ling, ``Almost universal codes for {MIMO}
  wiretap channels,'' \emph{{IEEE} Trans. Inf. Theory}, vol.~64, no.~11, pp.
  7218--7241, 2018.

\bibitem{Ling-Belfiore}
C.~Ling and J.-C. Belfiore, ``Achieving {AWGN} channel capacity with lattice
  {G}aussian coding,'' \emph{{IEEE} Trans. Inf. Theory}, vol.~60, no.~10, pp.
  5918--5929, 2014.

\bibitem{Campello-Dadush-Ling}
A.~Campello, D.~Dadush, and C.~Ling, ``{AWGN}-goodness is enough:
  {C}apacity-achieving lattice codes based on dithered probabilistic shaping,''
  \emph{{IEEE} Trans. Inf. Theory}, vol.~65, no.~3, pp. 1961--1971, 2018.

\bibitem{polar-lattices}
L.~Liu, Y.~Yan, and C.~Ling, ``Achieving secrecy capacity of the {G}aussian
  wiretap channel with polar lattices,'' \emph{{IEEE} Trans. Inf. Theory},
  vol.~64, no.~3, pp. 1647--1665, 2018.

\bibitem{Belfiore-Oggier}
J.-C. Belfiore and F.~Oggier, ``Lattice code design for the {R}ayleigh fading
  wiretap channel,'' in \emph{IEEE Int. Conf. on Commun. Workshops}, 2011, pp.
  1--5.

\bibitem{belfiore_mimo}
------, ``An error probability approach to {MIMO} wiretap channels,''
  \emph{{IEEE} Trans. Commun.}, vol.~61, no.~8, pp. 3396--3403, 2013.

\bibitem{mirghasemi}
H.~Mirghasemi and J.-C. Belfiore, ``Lattice code design criterion for {MIMO}
  wiretap channels,'' in \emph{Proc. {IEEE} Inf. Theory Workshop}, 2015, pp.
  277--281.

\bibitem{INS_me}
A.-M. Ernvall, D.~Karpuk, C.~Hollanti, and E.~Viterbo, ``Probability estimates
  for fading and wiretap channels from ideal class zeta functions,''
  \emph{Advances in Mathematics of Communications}, vol.~9, no.~4, pp.
  391--413, 2015.

\bibitem{INS_roope}
R.~Vehkalahti, H.~feng Lu, and L.~Luzzi, ``Inverse determinant sums and
  connections between fading channel information theory and algebra,''
  \emph{IEEE Transactions on Information Theory}, vol.~59, no.~9, pp.
  6060--6082, 2013.

\bibitem{Voronoi}
G.~Voronoi, ``Nouvelles applications des param\'etres continus \`{a} th\'eorie
  des formes quadratiques: 1. {Sur} quelques propri\'et\'es des formes
  quadratiques parfait,'' \emph{J. reine angew. Mat}, vol. 133, pp. 97--178,
  1908.

\bibitem{mcmullen}
C.~McMullen, ``Minkowski?s conjecture, well-rounded lattices and topological
  dimension,'' \emph{Journal of the American Mathematical Society}, vol.~18,
  no.~3, pp. 711--734, 2005.

\bibitem{ash}
A.~Ash, ``Small-dimensional classifying spaces for arithmetic subgroups of
  general linear groups,'' \emph{Duke Mathematical Journal}, vol.~51, no.~2,
  pp. 459--468, 1984.

\bibitem{solan2019stable}
O.~N. Solan, ``Stable and well-rounded lattices in diagonal orbits,''
  \emph{Israel Journal of Mathematics}, vol. 234, no.~2, pp. 501--519, 2019.

\bibitem{TaoLenny}
M.~T. Damir and L.~Fukshansky, ``Canonical basis twists of ideal lattices from
  real quadratic number fields,'' \emph{Houston Journal of Mathematics},
  vol.~45, pp. 999--1019, 2019.

\bibitem{micciancio}
D.~Micciancio, ``Generalized compact knapsacks, cyclic lattices, and efficient
  one-way functions from worst-case complexity assumptions,'' in \emph{The 43rd
  Annual IEEE Symposium on Foundations of Computer Science, 2002.
  Proceedings.}\hskip 1em plus 0.5em minus 0.4em\relax IEEE, 2002, pp.
  356--365.

\bibitem{luzzi_isit16}
L.~Luzzi, R.~Vehkalahti, and C.~Ling, ``Almost universal codes for fading
  wiretap channels,'' in \emph{IEEE Int. Symp. Inf. Theory}, 2016.

\bibitem{costawell}
R.~R. de~Araujo and S.~I. Costa, ``Well-rounded algebraic lattices in odd prime
  dimension,'' \emph{Archiv der Mathematik}, vol. 112, no.~2, pp. 139--148,
  2019.

\bibitem{damir2020bases}
M.~T. Damir and G.~Mantilla-Soler, ``Bases of minimal vectors in {L}agrangian
  lattices,'' \emph{arXiv preprint arXiv:2006.16794}, 2020.

\bibitem{Martinet_perfect}
J.~Martinet, \emph{Perfect Lattices in {Euclidean} Spaces}.\hskip 1em plus
  0.5em minus 0.4em\relax Springer, Berlin, 2010.

\bibitem{Schurmann}
A.~Schurmann, \emph{Computational geometry of positive definite quadratic
  forms: Polyhedral reduction theories, algorithms, and applications}.\hskip
  1em plus 0.5em minus 0.4em\relax American Mathematical Soc., 2009, vol.~48.

\bibitem{Shamai}
Y.~Liang, H.~V. Poor, S.~Shamai \emph{et~al.}, ``Information theoretic
  security,'' \emph{Foundations and Trends{\textregistered} in Communications
  and Information Theory}, vol.~5, no. 4--5, pp. 355--580, 2009.

\bibitem{Sloane}
J.~H. Conway and N.~J.~A. Sloane, \emph{Sphere packings, lattices and
  groups}.\hskip 1em plus 0.5em minus 0.4em\relax Springer Science \& Business
  Media, 2013, vol. 290.

\bibitem{Blichfield}
H.~Blichfeldt, ``The minimum values of positive quadratic forms in six, seven
  and eight variables,'' \emph{Mathematische Zeitschrift}, vol.~39, no.~1, pp.
  1--15, 1935.

\bibitem{Kumar}
H.~Cohn and A.~Kumar, ``Optimality and uniqueness of the {L}eech lattice among
  lattices,'' \emph{arXiv preprint math/0403263}, 2004.

\bibitem{PhongNguyen}
P.~Q. Nguyen and B.~Vall{\'e}e, \emph{The {LLL} algorithm}.\hskip 1em plus
  0.5em minus 0.4em\relax Springer, 2010.

\bibitem{Martinetbasis}
J.~Martinet, ``Bases of minimal vectors in lattices, {I},'' \emph{Archiv der
  Mathematik}, vol.~89, no.~5, pp. 404--410, 2007.

\bibitem{Viterbo}
F.~Oggier, E.~Viterbo \emph{et~al.}, ``Algebraic number theory and code design
  for {R}ayleigh fading channels,'' \emph{Foundations and
  Trends{\textregistered} in Communications and Information Theory}, vol.~1,
  no.~3, pp. 333--415, 2004.

\bibitem{Oggier_slow}
J.~Lu, J.~Harshan, and F.~Oggier, ``Performance of lattice coset codes on a
  {USRP} {T}estbed,'' \emph{arXiv preprint arXiV:1607.07163}, 2016.

\bibitem{Karrila-Hollanti}
A.~Karrila and C.~Hollanti, ``A comparison of skewed and orthogonal lattices in
  {G}aussian wiretap channels,'' in \emph{Proc. {IEEE} Inf. Theory Workshop},
  2015, pp. 1--5.

\bibitem{Amaro_approx}
A.~Barreal, M.~T. Damir, R.~Freij-Hollanti, and C.~Hollanti, ``An approximation
  of theta functions with applications to communications,'' \emph{arXiv
  preprint arXiv:1601.05596v3}, 2016/2020.

\bibitem{Chua}
K.~Chua, ``The root lattice ${A}_{n}^{\ast}$ and {R}amanujan?s circular
  summation of theta functions,'' \emph{Proceedings of the American
  Mathematical Society}, vol. 130, no.~1, pp. 1--8, 2002.

\bibitem{Peikert}
C.~Peikert, ``An efficient and parallel {G}aussian sampler for lattices,'' in
  \emph{Annual Cryptology Conference}, 2010, pp. 80--97.

\bibitem{Sarnak-Strombergsson}
P.~Sarnak and A.~Str{\"o}mbergsson, ``Minima of {E}pstein?s zeta function and
  heights of flat tori,'' \emph{Inventiones mathematicae}, vol. 165, no.~1, pp.
  115--151, 2006.

\bibitem{CostaDn}
G.~C. Jorge, A.~J. Ferrari, and S.~I. Costa, ``Rotated ${D}_n$-lattices,''
  \emph{Journal of Number Theory}, vol. 132, no.~11, pp. 2397--2406, 2012.

\bibitem{Jorge2}
G.~C. Jorge, A.~A. de~Andrade, S.~I. Costa, and J.~E. Strapasson, ``Algebraic
  constructions of densest lattices,'' \emph{Journal of Algebra}, vol. 429, pp.
  218--235, 2015.

\bibitem{coulangeon}
R.~Coulangeon, ``Spherical designs and zeta functions of lattices,''
  \emph{International Mathematics Research Notices}, vol. 2006, 2006.

\bibitem{bachoc}
C.~Bachoc, ``Designs, groups and lattices,'' \emph{Journal de th{\'e}orie des
  nombres de Bordeaux}, vol.~17, no.~1, pp. 25--44, 2005.

\bibitem{terras}
A.~Terras, ``The minima of quadratic forms and the behavior of {Epstein and
  Dedekind} zeta functions,'' \emph{Journal of Number Theory}, vol.~12, no.~2,
  pp. 258--272, 1980.

\bibitem{delone}
B.~N. Delone and S.~S. Ryshkov, ``A contribution to the theory of the extrema
  of a multi-dimensional $\zeta$-function,'' in \emph{Doklady Akademii Nauk},
  vol. 173, no.~5, 1967, pp. 991--994.

\bibitem{bayerNF}
E.~Bayer-Fluckiger, ``Lattices and number fields,'' \emph{Contemporary
  Mathematics}, 1999.

\bibitem{neukirch2013algebraic}
J.~Neukirch, \emph{Algebraic number theory}.\hskip 1em plus 0.5em minus
  0.4em\relax Springer Science \& Business Media, 2013, vol. 322.

\bibitem{fukshanskypet}
L.~Fukshansky and K.~Petersen, ``On well-rounded ideal lattices,''
  \emph{International Journal of Number Theory}, vol.~8, no.~01, pp. 189--206,
  2012.

\bibitem{bayersuarez}
E.~Bayer-Fluckiger and I.~Suarez, ``Ideal lattices over totally real number
  fields and {Euclidean} minima,'' \emph{Archiv der Mathematik}, vol.~86, pp.
  217--225, 2006.

\bibitem{Planewalker}
P.~Pyrr\"{o}, O.~Gnilke, C.~Hollanti, and M.~Greferath, ``Planewalker sphere
  decoder implementation,''
  \url{https://version.aalto.fi/gitlab/pasi.pyrro/sphere-decoder/tree/SPAWC},
  2018, [Online].

\bibitem{Viterbo_tables}
E.~Viterbo and F.~Oggier, ``Full diversity rotations,'' \url{
  http://www.ecse.monash.edu.au/staff/eviterbo/rotations/rotations.html}, 2005,
  [Online].

\bibitem{viterbo2}
F.~Oggier, E.~Bayer-Fluckiger, and E.~Viterbo, ``New algebraic constructions of
  rotated cubic lattice constellations for the {R}ayleigh fading channel,'' in
  \emph{Proc. {IEEE} Inf. Theory Workshop}, 2003, pp. 263--266.

\end{thebibliography}

\section*{Appendix}

\cami{
\subsection{Proof of the convergence of the series in \eqref{eq: def of psi series}}\label{convergence19}
Let $\Lambda$ be a lattice in $\R^n$. We consider 
\begin{align*}
f:\R^n&\rightarrow\R\\
t=(t_1,\dots,t_n)&\mapsto \prod_{i=1}^{n}\frac{1}{(1+t_i^2\frac{\sigma_h^2}{\sigma^2})^{3/2}}.
\end{align*}
In order to prove the convergence of the series in \eqref{eq: def of psi series}, we will compare $\sum_{t \in \Lambda} \vol(\Lambda) f(t)$ and   $\int_{x \in \R^n} f(x) dx$, using the observation $$\int_{x \in \R^n} f(x) dx=\sum_{t \in \Lambda} \int_{x \in V(t)} f(x) dx,$$
where $V(t)$ is the Voronoi cell around $t$. 
By explicit differentiation, the length of the logarithmic gradient ($f (x)'/f(x))$ is uniformly bounded by a constant $C = C(\sigma_h / \sigma) > 0.$

Thus, for $f(x)/f(t)$, where $x \in V(t)$, we have by the Intermediate value theorem 
$$\log [ f(x)/f(t) ] = \log f(x) - \log f(t) \geq - C d(x, t) \geq -C \diam(V)$$
and thus
$f(x) \geq \exp(-C \diam(V) ) f(t),$
where $\diam(V):=\max\{||x-y||~:~x,y\in V(t)\}$.

Note that $C$ and $\diam(V)$ above are uniform over all $t  \in \Lambda$ and $x \in V(t)$.
Hence,
$$\int_{x \in V(t)} f(x) dx
\geq \exp(-C \diam(V) ) \vol(\Lambda) f(t)$$
for all $t \in \Lambda$, and finally
$$\sum_{t \in \Lambda} \vol(\Lambda) f(t)
\leq \exp(C \diam(V) ) \int_{x \in \R^n } f(x) dx. $$

Clearly, the function $f$ is integrable on $\R^n$ and $f(0)=1$. Furthermore
$$\int_{\R^n}f(x)dx= \prod_{i=1}^{n}2\left[\frac{1}{\sqrt{1/t_i^2+\frac{\sigma_h^2}{\sigma^2}}}\right]_{0}^{+\infty}=\left(\frac{2\sigma}{\sigma_\mathbf{h}}\right)^n<\infty.$$
Hence, the convergence of the series in \eqref{eq: def of psi series}.

In fact, the average of $\psi_{\Lambda_e}\left(\frac{\sigma_\mathbf{h}}{\sigma}\right)$ over the set of determinant one lattices is $1+\left( \frac{\sigma_\mathbf{h}}{2 \sigma } \right)^n$. 
To prove this fact we use a celebrated theorem by Siegel given below.

We first identify the space of all $n$-dimensional (volume one) lattices with $$X_n:=\SL_n(\R)/\SL_n(\Z).$$ It is well known that $X_n$ has a unique (Haar) measure $\mu_n$.
\begin{thm}[Siegel's mean value theorem]
Suppose that $n\geq 2$. Let $f:\R^n\rightarrow \R$ be an integrable function. Let $\Lambda\in X_n$. Then
$$\int_{X_n}\sum_{t\in\Lambda}f(t)d\mu_n = \int_{\R^n}f(x)dx+f(0).$$
\end{thm}

Now, using the  above notation, we can conclude that the mean of $\psi_{\Lambda_e}\left(\frac{\sigma_\mathbf{h}}{\sigma}\right)$ over $X_n$ is $1+\left( \frac{\sigma_\mathbf{h}}{2 \sigma } \right)^n <\infty$. 
}

\subsection{Proof of Theorem \ref{thm:inf_bound_mod_lambda}}

\begin{proof}
By the independence assumptions, we have the identity
\begin{align}
    I\left[\mathbf{m}; (\mathbf{y}/\Lambda_{s, \mathbf{h}} , \mathbf{h})\right] = \E_\mathbf{h} \left[ I\left[\mathbf{m} ; (\mathbf{y}/\Lambda_{s, \mathbf{h}}  | \mathbf{h})\right]\right].
\end{align}
We divide $\mathbf{y}$ into components $\mathbf{y}_{\perp}$ and $\mathbf{y}_{\parallel}$, perpendicular and parallel to the nested lattices $\Lambda_{*, \mathbf{h}}$. Given $\mathbf{h}$, $\mathbf{y}_{\perp}$ consists only of the perpendicular noise component $\mathbf{n}_\perp$ and is hence independent of both $\mathbf{m}$ and $\mathbf{y}_{\parallel}$. Thus, $$I\left[\mathbf{m} ; (\mathbf{y}/\Lambda_{s, \mathbf{h}}  | \mathbf{h})\right] = I\left[\mathbf{m} ; (\mathbf{y}_{\parallel}/\Lambda_{s, \mathbf{h}}  | \mathbf{h})\right].$$
Next, given also the message $\mathbf{m}$, the variational distance of $\mathbf{y}_{\parallel}/\Lambda_{s, \mathbf{h}}$ and the uniform distribution on $\mathcal{V}(\Lambda_{s, \mathbf{h} }) $ can be bounded as in \cite{Viterbo-flatness}: the respective PDFs are
\begin{equation}
\rho_{ \{ \mathbf{y}_{\parallel}/\Lambda_{s, \mathbf{h}} | \mathbf{m} \} } (\mathbf{y}) = \frac{1}{[\Lambda_e : \Lambda_s]} g_n (\Lambda_{e, \mathbf{h}} - \lambda_m + \mathbf{y}, \sigma ),
\end{equation}
where $n = \dim (\mathrm{span}(\Lambda_*))$ is the rank of the nested lattices and $\mathbf{y} \in \mathcal{V}(\Lambda_{s, \mathbf{h} }) $, and
\begin{equation}
\rho_{ \mathrm{Unif} } (\mathbf{y}) = \frac{1}{[\Lambda_e : \Lambda_s]} \frac{1}{\vol (\Lambda_{e, \mathbf{h}})}.
\end{equation}
The definition of the flatness factor now directly implies
\begin{equation}
V (\rho_{ \{ \mathbf{y}_{\parallel}/\Lambda_{s, \mathbf{h}} | M = m \} }, \rho_{ \mathrm{Unif} }) \le \varepsilon_{\Lambda_{e, \mathbf{h}}} (\sigma).
\end{equation}
For simplicity, we denote $\varepsilon = \varepsilon_{\Lambda_{e, \mathbf{h}}} (\sigma)$ for the rest of this proof. For $\varepsilon \le 1/2$, Lemma~\ref{lemma: ff and information} now yields an information leakage bound $h(\varepsilon, | \mathcal{M} | )$, and otherwise we have the trivial upper bound $\log | \mathcal{M} |$. Hence,
\begin{align}
I\left[M; (\mathbf{y}/\Lambda_{s, \mathbf{h}} , \mathbf{h})\right] 
&= \E_{\mathbf{h}} [ I\left[M ; (\mathbf{y}_{\parallel}/\Lambda_{s, \mathbf{h}}  | \mathbf{h})\right] ] \\
&\le \E_\mathbf{h}\left[\mathbbm{1}_{\left\{\varepsilon \le 1 / 2\right\}} h(\varepsilon, | \mathcal{M} | ) \right] + \E_{\mathbf{h}}\left[\mathbbm{1}_{\left\{\varepsilon > 1/2\right\}} \log | \mathcal{M} | \right] \\
 \label{eq: two term info bd}
&= \mathbb{P}_{\mathbf{h}} \left[ \varepsilon \le 1/ 2 \right] \E_{\left\{\mathbf{h} | \varepsilon \le 1 / 2\right\}}\left[ h(\varepsilon, | \mathcal{M} | ) \right] + \mathbb{P}_{\mathbf{h}} \left[ \varepsilon >1 / 2 \right]  \log | \mathcal{M} |.
\end{align}

For the first term in \eqref{eq: two term info bd}, we apply Jensen's inequality to the convex function $h$ in $\varepsilon$,
\begin{align}
\E_{\left\{\mathbf{h} | \varepsilon \le 1 / 2\right\}}\left[ h(\varepsilon, | \mathcal{M} | )  \right] 
&\le h \left( \E_{\left\{ \mathbf{h} | \varepsilon \le 1 / 2\right\}}  \left[ \varepsilon \right], | \mathcal{M} | \right)  \\
   &\le h \left( \min \{   \E_{ \mathbf{h}}\left[\varepsilon \right] , 1/2 \}, | \mathcal{M} | \right).
\end{align}
The second inequality holds since $0 \le \E_{\left\{\mathbf{h} | \varepsilon \le 1 / 2\right\}}\left[\varepsilon\right] \le  \min \{   \E_{ \mathbf{h}}\left[\varepsilon \right] , 1/2 \} \le 1/2$, and $h$ is increasing on $[0, 1/2]$.

Next, write \eqref{eq: two term info bd} as a convex combination of two numbers,
\begin{align}
    I  \left[M; (\mathbf{y}/\Lambda_{s, \mathbf{h}} , \mathbf{h})\right] 
 \label{eq: two term info bd'}
&\le (1 - \mathbb{P}_{\mathbf{h}} \left[ \varepsilon >1 / 2 \right] ) h \left(  \min \{   \E_{ \mathbf{h}}\left[\varepsilon \right] , 1/2 \}, | \mathcal{M} | \right) + \mathbb{P}_{\mathbf{h}} \left[ \varepsilon >1 / 2 \right]  \log | \mathcal{M} |.
\end{align}
In the interval $[0, 1/2]$, we have $h(\cdot, | \mathcal{M} |) \le \log | \mathcal{M} |$, so the latter number in the convex combination \eqref{eq: two term info bd'} is the larger one. We can bound its weight using Markov's inequality to obtain $\mathbb{P}_{\mathbf{h}} \left[ \varepsilon >1 / 2 \right] \le 2 \E_{\mathbf{h}}\left[\varepsilon\right].$  Thus, we have
\begin{align}
I  \left[M; (\mathbf{y}/\Lambda_{s, \mathbf{h}} , \mathbf{h})\right] 
& \le (1- 2 \E_{\mathbf{h}}\left[\varepsilon\right]) h \left(  \min \{   \E_{ \mathbf{h}}\left[\varepsilon \right] , 1/2 \}, | \mathcal{M} | \right) + 2 \E_{\mathbf{h}}\left[\varepsilon\right] \log | \mathcal{M} | \\
& = \begin{cases} (1- 2 \E_{\mathbf{h}}\left[\varepsilon\right]) h \left( \E_{ \mathbf{h}}\left[\varepsilon \right], | \mathcal{M} | \right) + 2 \E_{\mathbf{h}}\left[\varepsilon\right] \log | \mathcal{M} |, \qquad \E_{ \mathbf{h}}\left[\varepsilon \right] \le 1/2 \\
\log | \mathcal{M} |, \qquad \E_{ \mathbf{h}}\left[\varepsilon \right] \ge 1/2. \end{cases}
\end{align}
The theorem follows.
\end{proof}

\subsection{Proof of Lemma \ref{lem:discrete_continuous_gaussian} }

\begin{proof}
We start with a technical modification of $\rho (\mathbf{y})$. By construction, $\mathbf{y} = \mathbf{h} \mathbf{x} + \mathbf{n}$ has the PDF
\begin{align}
 \rho (\mathbf{y}) &= \sum_{\mathbf{x} \in \Lambda_e + \boldsymbol{\lambda}_M } P(\mathbf{x} = \mathbf{x}) \rho_{\mathbf{y} | \mathbf{x} = \mathbf{x}}(\mathbf{y} ) \\
&= \frac{1}{g_n(\Lambda_e + \boldsymbol{\lambda}_M ; \sigma_s)} \sum_{\mathbf{x} \in \Lambda_e + \boldsymbol{\lambda}_M } g_n(\mathbf{x}; \sigma_s) g_m(\mathbf{y} -  \mathbf{h} \mathbf{x}; \sigma ) \\
\label{eq: PDF modification}
&=  \frac{1}{g_m (\Lambda_e + \boldsymbol{\lambda}_M ; \sigma_s)} \frac{1}{\sqrt{2 \pi \sigma_s}^m \sqrt{2 \pi \sigma}^m } \sum_{\mathbf{x} \in \Lambda_e + \boldsymbol{\lambda}_M } \exp \left[ -\frac{1}{2 \sigma^2 \sigma_s^2 } \left( \sigma^2 | \mathbf{x} |^2 + \sigma^2 | \mathbf{y} - \mathbf{hx} |^2 \right) \right].
\end{align}
Let us expand separately the norms in the exponential:
\begin{equation}
\sigma^2 | \mathbf{x} |^2 + \sigma_s^2 | \mathbf{y} - \mathbf{hx} |^2 
= \sigma^2 \mathbf{x}^t \mathbf{x} +  \sigma_s^2 \mathbf{x}^t \mathbf{h}^t \mathbf{h}  \mathbf{x}  - \sigma_s^2 (\mathbf{y}^t \mathbf{h x} + \mathbf{x}^t \mathbf{h}^t \mathbf{y} ) + \sigma_s^2 | \mathbf{y}  |^2
\end{equation}
Notice that $(\sigma^2 I_{n} + \sigma_s^2 \mathbf{h}^t \mathbf{h})$ is a positive definite symmetric matrix. Let $\mathbf{Q} \in \R^{n \times n}$ be its square-root matrix, $(\sigma^2 I_{n} + \sigma_s^2 \mathbf{h}^t \mathbf{h}) = \mathbf{Q}^t \mathbf{Q}$. Note that $\mathbf{Q}$ is invertible since $\ker (\mathbf{Q}^t \mathbf{Q}) = \{ \mathbf{0} \}$. 
A straightforward calculation yields
\begin{align}\label{eq: matrix computation 1}
\sigma^2 | \mathbf{x} |^2 + \sigma_s^2 | \mathbf{y} - \mathbf{hx} |^2 
&= | \mathbf{Q x} - \sigma_s^2 \mathbf{Q}^{-t} \mathbf{h}^{t}  \mathbf{y}  |^2 + \sigma_s^2  \mathbf{y}^t (I_{m}- \sigma_s^2 \mathbf{hQ}^{-1} \mathbf{Q}^{-t} \mathbf{h}^{t}  ) \mathbf{y} \\
&=| \mathbf{Q x} - \sigma_s^2  \mathbf{Q}^{-t} \mathbf{h}^{t}  \mathbf{y}  |^2 + \sigma_s^2  \sigma^2 \mathbf{y}^t (\sigma^2 I_{m} + \sigma_s^2 \mathbf{h h }^t )^{-1} \mathbf{y}.
\end{align}
We substitute this back into \eqref{eq: PDF modification} to obtain
\begin{align}
 \rho (\mathbf{y}) 
&= \frac{1}{g_m (\Lambda_e + \boldsymbol{\lambda}_M ; \sigma_s)} \frac{1}{\sqrt{2 \pi \sigma_s}^m \sqrt{2 \pi \sigma}^m } \exp ( - \frac{1}{2}  \mathbf{y}^t  (\sigma^2 I_{m} + \sigma_s^2 \mathbf{h h}^t )^{-1} \mathbf{y} ) \\ 
&\qquad \times\sum_{\mathbf{x} \in \Lambda_e + \boldsymbol{\lambda}_M } \exp \left[ -\frac{1}{2 \sigma^2 \sigma_s^2 } | \mathbf{Q x} -   \sigma_s^2  \mathbf{Q}^{-t} \mathbf{h}^{t}  \mathbf{y}  |^2  \right] \\
&=  \frac{1}{g_m (\Lambda_e + \boldsymbol{\lambda}_M ; \sigma_s)} \frac{ \sqrt { \det (\sigma^2 I_{m} + \sigma_s^2 \mathbf{hh}^t ) } }{( \sqrt{2 \pi }  \sigma_s \sigma )^m } \tilde{\rho} (\mathbf{y}) \sum_{\mathbf{x} \in \Lambda_e + \boldsymbol{\lambda}_M } \exp \left[ -\frac{1}{2 \sigma^2 \sigma_s^2 } | \mathbf{Q x} -  \sigma_s^2  \mathbf{Q}^{-t} \mathbf{h}^{t}  \mathbf{y} |^2  \right] 
\end{align}
Let $\mathbf{h} = U D V$ be the singular value decomposition of $\mathbf{h}$, where $U \in \R^{m \times m}$ and $V\in \R^{n \times n}$ are orthonormal and $D\in \R^{m \times n}$ is a (possibly nonsquare) diagonal matrix with diagonal entries $d_1, \ldots, d_n$. Then, we have
\begin{equation}
\det(\mathbf{Q}^t \mathbf{Q}) = \det(V^t(\sigma^2 I_{n} + \sigma_s^2 D^t D)V)  = \prod_{i=1}^n( \sigma^2 + \sigma_s^2 d_i^2 ),
\end{equation}
and similarly
\begin{equation}
\det (\sigma^2 I_{m} + \sigma_s^2 \mathbf{hh}^t ) = \det(U^t (\sigma^2 I_{m} + \sigma_s^2 DD^t )  U) = \sigma^{2(m-n)} \det(\mathbf{Q}^t \mathbf{Q}).
\end{equation}
Hence,
\begin{align}
\rho (\mathbf{y}) &= \tilde{\rho} (\mathbf{y})  \frac{ \sqrt { \det (\sigma^2 I_{m} + \sigma_s^2 \mathbf{hh}^t ) } \vol (\Lambda_e) }{\vol (\Lambda_e) g_m (\Lambda_e + \boldsymbol{\lambda}_M ; \sigma_s)}   \sum_{\mathbf{t} \in  Q \Lambda_e + \mathbf{u} } g_m (\mathbf{t}; \sigma_s \sigma) \\
\label{eq: y PDF}
&= \tilde{\rho} (\mathbf{y})  \frac{ \vol (\frac{1}{\sigma} \mathbf{Q}  \Lambda_e)  g_n (\frac{1}{\sigma}  \mathbf{Q}  \Lambda_e + \frac{1}{\sigma}  \mathbf{u}; \sigma_s ) }{\vol (\Lambda_e) g_n (\Lambda_e + \boldsymbol{\lambda}_M ; \sigma_s)}
\end{align}
where $\mathbf{u}$ is a suitable vector. This form of the PDF $\rho (\mathbf{y})$ allows us to bound the variational distance to $\tilde{\rho} (\mathbf{y}) $.


Let us study the latter factor in \eqref{eq: y PDF}.  Since the flatness factor is rotationally invariant, we may write everything in terms of an eigenbasis of $\mathbf{h}^t \mathbf{h}$. Since $\mathbf{h}^t \mathbf{h}$ is symmetric and positive semi-definite, the basis is orthonormal and in this basis $\mathbf{h}^t \mathbf{h}= \diag (h_i^2)$, and $\mathbf{Q} = \frac{1}{\sigma} \sqrt{ \sigma^2 I_{n} + \sigma_s^2 \mathbf{h}^t \mathbf{h} } = \diag(\sqrt{ 1+ \sigma_s^2 h_i^2 / \sigma^2 } )$. We have
\begin{align}
\left| \vol \left(\frac{1}{\sigma} \mathbf{Q}  \Lambda_e\right)  g_n \left(\frac{1}{\sigma}  \mathbf{Q}  \Lambda_e + \frac{1}{\sigma}  \mathbf{u}; \sigma_s \right) - 1 \right| &\le 
 \varepsilon_{ \diag(\tilde{h}) \Lambda_e  } (\sigma_s) 
 \\
&= \Theta_{(\diag(\tilde{h}) \Lambda_e)^*} \left(e^{- 2 \pi  \sigma_s^2} \right) - 1 \\
&= \sum_{\mathbf{t} \in \Lambda_e^*} \exp \left( -  2 \pi \sigma^2 \sigma_s^2 \sum_{i=1}^n \frac{t_i^2}{h_i^2 \sigma_s^2 + \sigma^2}\right) -1 \\
\label{eq: first theta series}
&= \sum_{\mathbf{t} \in \Lambda_e^*} \exp \left( -  2 \pi \sum_{i=1}^n \frac{t_i^2}{h_i^2/ \sigma^2 + 1/ \sigma_s^2}\right) - 1.
\end{align}
Similarly, the denominator of the latter factor in \eqref{eq: y PDF} satisfies
\begin{equation}\label{eq: second theta series}
| \vol( \Lambda_e ) g_n( \Lambda_e  ; \sigma_s  ) - 1 | 
 \le  \varepsilon_{ \diag(\tilde{h}) \Lambda_e  } (\sigma_s ) = \sum_{\mathbf{t} \in \Lambda_e^*} \exp \left( -  2 \pi \sum_{i=1}^n \frac{t_i^2}{1/ \sigma_s^2}\right) - 1.
\end{equation}
From expressions \eqref{eq: first theta series} and \eqref{eq: second theta series} it is also clear that $\varepsilon_{ \diag(\tilde{h}) \Lambda_e  } (\sigma_s \sigma) \ge \varepsilon_{  \Lambda_e  } (\sigma_s ) $. Hence, the latter factor in \eqref{eq: y PDF} is between $\frac{1- \varepsilon}{1 + \varepsilon}$ and $\frac{1 + \varepsilon}{1 - \varepsilon}$, where $$\varepsilon = \varepsilon_{ \frac{1}{\sigma} Q \Lambda_e } (\sigma_s ) =\varepsilon_{ \frac{1}{\sigma_s} Q \Lambda_e  } (\sigma) = \varepsilon_{\sqrt{\sigma^2/\sigma_s^2 I_{n} +  \mathbf{h}^t\mathbf{h}}\Lambda_e}(\sigma),$$using the scaling property. Consequently,  the deviation of the latter factor in \eqref{eq: y PDF} from $1$ is at most
\begin{equation}
\frac{1 + \varepsilon}{1 - \varepsilon} - 1= \frac{2\varepsilon}{1 - \varepsilon} \le \frac{ 2 \varepsilon}{1- \varepsilon_{\mathrm{max} } },
\end{equation}
where we used the assumption $\varepsilon \le \varepsilon_{\mathrm{max} }$. Thus,
\begin{equation}
| \rho(\mathbf{y}) - \tilde{\rho} (\mathbf{y}) | \le \frac{ 2 \varepsilon}{1- \varepsilon_{\mathrm{max} } } \tilde{\rho} (\mathbf{y}) 
\end{equation}
and integrating over $\R^n$ we get the proposed statistical distance.
\end{proof}

\subsection{Proof of Theorem \ref{thm: info bd}}
\begin{proof}
The proof closely follows the steps of that of Theorem~\ref{thm:inf_bound_mod_lambda}. We start by writing
\begin{align}
    I\left[M; (\mathbf{y} , \mathbf{h})\right] = \E_\mathbf{h} \left[ I\left[M ; (\mathbf{y} | \mathbf{h}=\mathbf{h})\right] \right].
\end{align}
For a fixed channel realization $\mathbf{h}$, by Lemma~\ref{lem:discrete_continuous_gaussian} the distribution of the received vector $\mathbf{y}$ is close to a fixed Gaussian distribution $\tilde{\rho}$ for all messages $M$, with variational distance
\begin{align}
    V\left(\rho_{\left\{\mathbf{y}| M=m\right\}},\tilde{\rho}\right) \le 2 \varepsilon_{\sqrt{\sigma^2/\sigma_s^2 I_{n} +  \mathbf{h}^t \mathbf{h}} \Lambda_e}(\sigma) / (1 - \varepsilon_{\mathrm{max} }),
\end{align}
provided that $\varepsilon_{\sqrt{\sigma^2/\sigma_s^2 I_{n} +  \mathbf{h}^t \mathbf{h}} \Lambda_e}(\sigma) =: \varepsilon \le \varepsilon_{\mathrm{max} }$ for some fixed $ \varepsilon_{\mathrm{max} } <1$.
Using Lemma~\ref{lemma: ff and information}, we get an information leakage bound for channel matrices $\mathbf{h}$ such that $2 \varepsilon / (1 - \varepsilon_{\mathrm{max} }) \le 1/2 $
\begin{equation}
I\left[M ; (\mathbf{y} | \mathbf{h}=\mathbf{h})\right]  \le h \left( \frac{2 \varepsilon }{1 - \varepsilon_{\mathrm{max} }}, | \mathcal{M} | \right).
\end{equation}
This bound is applicable for a largest range of values of $\varepsilon$ if we have $2 \varepsilon_{\mathrm{max} } / (1 - \varepsilon_{\mathrm{max} }) = 1/2$, so we choose $\varepsilon_{\mathrm{max}} = 1/5$. Otherwise, we have the trivial information leakage bound $\log | \mathcal{M} |$. This yields
\begin{align}
\label{eq: two-term bound}
   I \left[M; (\mathbf{y} , \mathbf{h})\right] & \le 
   \E_{ \mathbf{h} }  \left[  \mathbbm{1}_{ \left\{\varepsilon \le 1/5 \right\} }
    h \left( 5 \varepsilon / 2, | \mathcal{M} | \right)  \right] 
    +\E_{\mathbf{h}}\left[\mathbbm{1}_{\left\{ \varepsilon > 1/5  \right\}} \log | \mathcal{M} | \right].
\end{align}
The rest of the proof is identical to Theorem~\ref{thm:inf_bound_mod_lambda}: we use Jensen's inequality for the first term, 
\begin{align}
\E_{ \mathbf{h} }   \left[ \mathbbm{1}_{ \left\{\varepsilon \le 1/5 \right\} }
    h \left( 5 \varepsilon / 2, | \mathcal{M} | \right)  \right] 
    &= \mathbb{P} [ \varepsilon \le 1/5 ] \ \E_{ \mathbf{h} | \varepsilon \le 1/5 }
    \left[ h \left( 5 \varepsilon / 2, | \mathcal{M} | \right)  \right] \\
    & \le \mathbb{P} [ \varepsilon \le 1/5 ] \  h \left( 5 \min \{ \E_{\mathbf{h}} [\varepsilon ], 1/5 \} / 2, | \mathcal{M} | \right)
\end{align}
and increase the weight of the latter larger term in the convex combination \eqref{eq: two-term bound} by Markov's inequality,
\begin{align}
   I \left[M; (\mathbf{y} , \mathbf{h})\right] & \le 
    (1- 5 \E_{\mathbf{h}} [\varepsilon ]) h \left( 5 \min \{ \E_{\mathbf{h}} [\varepsilon ], 1/5 \} / 2, | \mathcal{M} | \right) + 5 \E_{\mathbf{h}} [\varepsilon ] \log | \mathcal{M} | \\
    & = \begin{cases} (1- 5 \E_{\mathbf{h}}\left[\varepsilon\right]) h \left(5 \E_{ \mathbf{h}}\left[ \varepsilon \right]/ 2, | \mathcal{M} | \right) + 5 \E_{\mathbf{h}}\left[\varepsilon\right] \log | \mathcal{M} |, \qquad \E_{ \mathbf{h}}\left[\varepsilon \right] \le 1/5 \\
\log | \mathcal{M} |, \qquad \E_{ \mathbf{h}}\left[\varepsilon \right] \ge 1/5. \end{cases}
\end{align}
which concludes the proof.
\end{proof}

\end{document}